\documentclass[12pt,letterpaper]{article}

\usepackage{amssymb,amsmath,amsfonts,eurosym,geometry,graphicx,caption,setspace,comment,footmisc,caption,natbib,pdflscape,subfigure,color} 

\usepackage{pdfpages}

\allowdisplaybreaks

\usepackage{kbordermatrix}

\usepackage{aliascnt}

\usepackage[
  bookmarks=true,
  bookmarksopen=true,
  breaklinks=true,
  colorlinks=true,
  citecolor=blue,
  linkcolor=blue,
  urlcolor=blue,
]{hyperref}

\usepackage{tikz}
\usetikzlibrary{shapes}
\usetikzlibrary{trees} 
\usetikzlibrary{patterns.meta}

\usepackage{forest}
\usetikzlibrary{angles}

\newtheorem{theorem}{Theorem}
\newtheorem{proposition}{Proposition}

\newtheorem{lemma}{Lemma}

\usepackage{aliascnt}
\newtheorem{claim}{Claim}
\newenvironment{proof}[1][Proof]{\noindent\textbf{#1.} }{\ \rule{0.5em}{0.5em}}

\makeatletter

\renewcommand\tagform@[1]{\maketag@@@{\ignorespaces#1\unskip\@@italiccorr}}
\renewcommand\theequation{(\oldtheequation)}
\makeatother

\newcommand{\ba}{{\mathbf a}}
\newcommand{\bb}{{\mathbf b}}

\newcommand{\bm}{{\mathbf m}}
\newcommand{\bp}{{\mathbf p}}
\newcommand{\bq}{{\mathbf q}}

\newcommand{\bv}{{\mathbf v}}
\newcommand{\bw}{{\mathbf w}}
\newcommand{\bx}{{\mathbf x}}

\newcommand{\bz}{{\mathbf z}}

\newcommand{\bA}{{\mathbf A}}
\newcommand{\bB}{{\mathbf B}}
\newcommand{\bC}{{\mathbf C}}

\newcommand{\bF}{{\mathbf F}}
\newcommand{\bG}{{\mathbf G}}
\newcommand{\bH}{{\mathbf H}}
\newcommand{\bI}{{\mathbf I}}
\newcommand{\bK}{{\mathbf K}}
\newcommand{\bM}{{\mathbf M}}

\newcommand{\bU}{{\mathbf U}}

\newcommand{\bX}{{\mathbf X}}

\newcommand{\bBbar}{\overline{\bB}}

\newcommand{\xbar}{\overline{\bx}}
\newcommand{\xop}{\bx^*}

\newcommand{\calB}{\mathcal{B}}

\newcommand{\R}{{\mathbb R}}

\newcommand{\cB}{{\mathcal B}}

\newcommand{\nn}{\nonumber}

\DeclareMathOperator{\E}{\mathrm{E}}
\DeclareMathOperator{\Var}{\mathrm{Var}}

\newcommand*{\trans}{\mathsf{T}}

\usepackage{etoolbox}
\newcommand{\zerodisplayskips}{%
  \setlength{\abovedisplayskip}{0.05 in}
  \setlength{\belowdisplayskip}{0.05 in}
  \setlength{\abovedisplayshortskip}{0.05 in}
  \setlength{\belowdisplayshortskip}{0.05 in}}
\appto{\normalsize}{\zerodisplayskips}
\appto{\small}{\zerodisplayskips}
\appto{\footnotesize}{\zerodisplayskips}

\definecolor{darkspringgreen}{rgb}{0.09, 0.45, 0.27} 
\definecolor{darkgray}{rgb}{0.66, 0.66, 0.66}

\usepackage{tikz}
\usetikzlibrary{arrows.meta,automata,positioning,quotes}

\usepackage{kotex}

\begin{document}
  

\title{Robust Intervention in Networks\thanks{We have benefited from conversations with Wonki Cho, Sung-Ha Hwang, Donggyu Kim, Duk Gyoo Kim, Jeong-Yoo Kim, Youngwoo Koh, Matthew Kovach, Omer Tamuz, and participants at several faculty seminars and conferences. Author names are listed in alphabetical order, and all the authors contributed equally to this paper.}
}


\author{
Daeyoung Jeong \and Tongseok Lim \and Euncheol Shin\thanks{Daeyoung Jeong: School of Economics, Yonsei University, Seoul, Republic of Korea. Email: {\tt daeyoung.jeong@gmail.com}; Tongseok Lim: Mitchell E. Daniels, Jr. School of Business, Purdue University, IN, USA. Email: {\tt lim336@purdue.edu}; Eunchol Shin: KAIST College of Business, Seoul, Republic of Korea. Email: {\tt eshin.econ@kaist.ac.kr}.}
}


\maketitle

\sloppy 

\singlespacing

\vspace{-0.3 in}

\begin{abstract}


In economic settings such as learning, social behavior, and financial contagion, agents interact through interdependent networks. This paper examines how a decision maker (DM) can design an optimal intervention strategy under network uncertainty, modeled as a zero-sum game against an adversarial ``Nature'' that reconfigures the network within an uncertainty set. Using duality, we characterize the DM's unique robust intervention and identify the worst-case network structure, which exhibits a rank-1 property, concentrating risk along the intervention strategy. We analyze the costs of robustness, distinguishing between global and local uncertainty, and examine the role of higher-order uncertainties in shaping intervention outcomes. Our findings highlight key trade-offs between maximizing influence and mitigating uncertainty, offering insights into robust decision-making. This framework has applications in policy design, economic regulation, and strategic interventions in dynamic networks, ensuring their resilience against uncertainty in network structures.
\end{abstract}

\vspace{-0.3 in}

\strut



\strut

\textbf{Keywords:} adversarial Nature; duality; network uncertainty; robust optimization

\onehalfspacing



\section{Introduction} \label{sec:introduction}


Central authorities must distribute limited resources across complex and uncertain networks, where decisions are interconnected and dynamic. Network theory can assist in the analysis of these systems, but the network structure is often only partially known, complicating decision-making. Once resources are allocated, participants may engage in actions such as trading, sharing, or forming collaborations, which can lead to unpredictable outcomes. Although recent literature has focused on intervention in complex networks \citep[e.g.,][]{Galeottietal:2020:ECMA,Galeottietal:2024b:WP,JeongShin2024, LiandTan:2025:WP, PariseandOzdaglar:2023:ECMA, Sunetal:IER:2023}, the role of uncertainty has not been fully addressed. Understanding how uncertainty affects intervention decisions is essential for improving the effectiveness and robustness of resource distribution plans.

In this paper, we address the challenge of designing robust intervention strategies in uncertain networks, where the structure and inter-dependencies are only partially known. Using a robust optimization framework, we model the problem as a zero-sum game between a decision maker (henceforth, DM) and an adversarial ``Nature.'' The DM seeks to allocate resources effectively to guide network outcomes toward specific targets, while Nature embodies uncertainty in the network, manipulating its structure to maximize the DM's objective function in the least favorable way. By analyzing the worst-case scenarios that emerge under this framework, we characterize the DM's optimal intervention strategy, which accounts for both the mean influence of the network and the risks introduced by uncertainty. 

Consider the allocation of medical supplies during a pandemic or the injection of liquidity during a financial crisis---both scenarios involve central authorities distributing limited resources across interconnected and dynamic networks. During a pandemic, authorities must allocate essential medical resources, such as vaccines, ventilators, or protective equipment, to various regions with specific targets, such as vaccinating a set percentage of the population to achieve herd immunity or ensuring that hospitals are adequately equipped to handle surges in patients.\footnote{\cite{cunningham2021pharmacies} examines logistical challenges in COVID-19 vaccine distribution, highlighting delays and resource misalignments that parallel the network uncertainties in this paper. 
} Similarly, during a financial crisis, central banks and funding bodies provide liquidity to commercial banks or allocate resources to institutions to achieve specific objectives, such as stabilizing credit markets or encouraging innovation through research funding.\footnote{\cite{liang2018lessons} highlights lessons from previous financial crises, emphasizing the role of central authorities in addressing systemic risks and uncertainties, which parallels the decision-making challenges under ambiguity discussed in our paper.}

However, in both cases, the interconnected nature of the underlying networks creates significant uncertainties. Regions or institutions often trade, share, or collaborate to address local fluctuations in needs or to capitalize on unexpected partnerships. For instance, during a pandemic, patients might be transferred between hospitals, or vaccines might be reallocated across regions experiencing sudden outbreaks. In financial networks, banks might redistribute liquidity through interbank lending or adjust their portfolios based on market conditions. These interactions, while potentially beneficial, can also lead to unpredictable spillover effects and deviations from the original allocation plans.

Such uncertainties pose significant challenges to the effectiveness of resource distribution. Unforeseen trading patterns, unexpected collaborations, and fluctuating demands can create imbalances, leaving some regions or institutions under-resourced while others receive excess supplies. To address these challenges, central authorities must incorporate potential uncertainties into their planning and adopt robust intervention strategies that ensure effective outcomes, even in the face of unpredictable changes and complex interdependencies.

These scenarios illustrate a fundamental problem faced by central authorities: how to intervene in uncertain networks to guide resource distribution toward specific goals. After the initial allocation, agents may engage in trading or collaboration, creating unpredictable outcomes. Central authorities need to incorporate potential network uncertainties into their strategies to ensure that final outcomes align with intended targets. In particular, they need intervention strategies that are robust in the face of such uncertainties, ensuring that final outcomes align with their objectives even when the network behaves unexpectedly. 

To tackle the challenge of robust intervention in uncertain networks, we frame the decision-making problem as a zero-sum game between a DM and an adversarial Nature. At the core of this analysis is the DM's objective function, which features a quadratic form that captures the interaction between the intervention strategy and the network’s uncertain structure. Nature's role is to manipulate the network by maximizing this quadratic term, effectively amplifying the uncertainty in the most detrimental way for the DM. Therefore, from the DM's point of view, by considering the worst-case scenarios resulting from these uncertainties, the DM devises intervention strategies that remain effective regardless of how the adversarial Nature manipulates the network formation process (\autoref{theorem:duality}).

We find that under regularity conditions, there is a unique worst-case scenario where Nature optimizes the correlation between agents' outcomes based on the signs of the intervention vector entries chosen by the DM (\autoref{thm:rank-1}). Importantly, Nature's chosen matrix exhibits a rank-1 property, meaning that uncertainty is concentrated in a single principal component aligned with the DM's intervention vector. This allows Nature to amplify the uncertainty in a specific direction, maximizing the potential harm to the DM's objective. The worst-case scenario is thus characterized by Nature selecting a principal component that focuses all uncertainty on the most damaging aspect of the DM's intervention. 

The consideration of network uncertainty and robust intervention introduces new insights into optimal intervention strategies. Previous models have focused on maximizing mean influence by targeting central agents, often overlooking the risks posed by uncertainty. Our analysis highlights the importance of balancing influence and uncertainty for effective interventions. In Section~\ref{sec:definitions}, we define the cost of uncertainty, distinguishing between global uncertainty, which concentrates risk along the intervention strategy, and local uncertainty, which propagates through network interactions. Section~\ref{subsec:roleofuncertainty} illustrates these effects using a two-agent network, showing that when uncertainty at a high-influence agent exceeds a threshold, the DM optimally reallocates resources to a lower-influence but more stable agent. This result highlights a key trade-off: while targeting high-influence agents is beneficial, excessive uncertainty can undermine intervention effectiveness, emphasizing the need for robust strategies that account for both influence and risk.

In Section~\ref{sec:applicationsandextensions}, we extend our analysis in two key directions. First, we incorporate higher-order interactions, which arise in examples such as network games, supply chains, and financial contagion models. Using a second-order approximation, we show that the rank-1 property of the worst-case scenario persists even when network effects propagate beyond direct interactions, ensuring that intervention strategies remain predictable. Second, we examine robust intervention strategies under partial information, focusing on scenarios where a new agent joins an existing network. This reflects real-world challenges, such as allocating resources in evolving social, research, or economic networks when interactions with new participants remain uncertain. We show that even under incomplete information, a unique worst-case scenario emerges, optimally balancing the DM's knowledge of the existing network with the uncertainty introduced by the new agent.

Taken together, these findings provide a generalizable framework for robust intervention in uncertain networks, offering insights applicable to policy-making, economic regulation, and resource allocation in dynamic environments. Our results emphasize that optimal interventions must not only target high-influence agents but also anticipate and mitigate uncertainty, ensuring more stable and resilient decision-making in complex network settings.

\subsection{Related Literature}\label{subsec:literature}

The current paper is broadly related to three strands of literature: network interventions, strategic interactions in networks, and robust mechanism design.

\vskip+1em

\noindent \textbf{Intervention in networks.} In the recent growing literature on intervention in networks, there are two approaches: (i) intervening in the network structure \citep[e.g.,][]{BDellalena:EER:2024, Sunetal:IER:2023} and (ii) intervening in incentives within a given network structure \citep[e.g.,][]{Belhajetal:GEB:2023, Galeottietal:2020:ECMA, Galeottietal:2024b:WP, JeongShin2024, LiandTan:2025:WP, PariseandOzdaglar:2023:ECMA}. Our paper relates to both strands of this literature. Specifically, we study how a DM can intervene in agents' strategic incentives while accounting for uncertainty in the realized network, which arises from a game-like interaction between the DM and an adversarial Nature. Under certain constraints, the network structure is determined by Nature, an adversarial player whose objective opposes that of the DM. \citet{Galeottietal:2024b:WP} also examine robust interventions to improve market efficiency in the context of oligopolistic market competition. Their notion of robustness focuses on achieving a certain property with high probability. In contrast, our approach emphasizes worst-case scenario optimization, aligning more closely with the robust mechanism design literature \citep[e.g.,][]{BirgemanandMorris:ECMA:2005, Carroll:AER:2015}.

\vskip+1em

\noindent \textbf{Strategic interactions in networks.} Intervention models are applicable to various frameworks, such as public goods games \citep[e.g.,][]{Allouch:JET:2015,BramoulleandKranton:JET:2007,GaleottiandGoyal:RAND:2009} and social learning models \citep[e.g.,][]{DeGroot1974reaching, DeMarzoetal:2003:QJE, GolubJackson2010, GolubandJackson:2012:QJE}. One particular area of the literature related to our research focuses on network games with uncertainty. Previous studies \citep[e.g.,][]{Chaudurietal;WP;2023,Galeotti:REStud:2010,Shin:MSS:2021} examine uncertainty on the agents' side, where agents make equilibrium decisions based on incomplete information about the underlying networks. In contrast, the current paper considers uncertainty on the DM's side; the DM, who intervenes in agents' behavior, has limited information about the underlying networks. 

\vskip+1em

\noindent \textbf{Robust mechanism design.} Our study relates to the growing literature on robust mechanism design. \citet{BirgemanandMorris:ECMA:2005} develop a robust implementation framework that achieves equilibrium under minimal assumptions about agents' knowledge. In a principal-agent model, \citet{Carroll:AER:2015} extends this approach by relaxing distributional assumptions and designing mechanisms that perform optimally under worst-case scenarios of informational uncertainty. Similar worst-case scenario and $\max \min$ approaches have been extensively applied not only in principal-agent settings \citep[e.g.,][]{Frankel:AER:2014, Garrett:GEB:2014, Kambhampati:TE:2024} but also in auction settings \citep[e.g.,][]{BrooksandDu:ECMA:2024,BrooksandDu:WP:2021, che2022distributionallyrobustoptimalauction, HeandLi:JET:2022, Kocyigit:2019:MS} and industrial organization theory \citep[e.g.,][]{GuoandShmaya:2025:AER}. 

To the best of our knowledge, the current paper is the first to study the robustness of a DM's intervention when there is uncertainty about the relevant network structure.\footnote{Several studies highlighted that network formation processes are often influenced by the surrounding environment. For instance, during COVID-19, \citet{Zhang2020} demonstrate that contact patterns between age groups shifted due to social distancing measures, creating uncertainty around susceptibility to infection and the dynamics of COVID-19. Similarly, \citet{Kangetal:2023:WP} present that the pandemic altered traffic patterns across New York City's five boroughs, as social distancing disrupted connectivity. In the context of adolescent friendship networks, \citet{Choietal:WP:2024} find that homophilistic friendship formation patterns change when educational pedagogy undergoes changes.} As in the previous robust mechanism design literature, we analyze the DM's intervention strategy in a network setting to ensure ex-ante optimal outcomes under worst-case scenarios. In our model, due to the quadratic objective function of the DM, the worst-case scenario is represented by the correlation structure, which is in line with, for example, \citet{CremerandMcLean:ECMA:1988}, \citet{HeandLi:JET:2022}, \citet{Myerson:MOR:1981}. In terms of modeling uncertainty, \citet{che2022distributionallyrobustoptimalauction} and \citet{Kocyigit:2019:MS} consider distributional robustness; in contrast, the current model approaches uncertainty as a matrix completion problem.




\section{Model}
\label{sec:model}


\noindent \textbf{Notation.} Throughout this paper, we represent each vector as a column vector. For a given vector $\ba$, we denote its $i$'th entry by $a_i$. For any two vectors $\ba$ and $\bb$, their inner product is denoted by $\langle \ba, \bb \rangle$, and their outer product is denoted by $\ba \otimes \bb$. $\| \ba \|$ represents the vector norm induced by the canonical inner product, specifically the $l^2$ norm. $s(a_i)$ represents the sign of each non-zero entry $a_i$ of vector $\ba$. Matrices are denoted in boldface, and for a matrix $\bA$, we use the following notation: $\bA_i$ represents its $i$th row (as a vector), $\bA^j$ represents its $j$th column (as a vector), $\bA_{ij}$ denotes the element in the $i$th row and $j$th column, and $\bA^\trans$ denotes its transpose. $\| \bA \|$ is the matrix norm induced by the vector norm, which itself is induced by the canonical inner product. $\| \bA \|_{\bF}$ denotes the Frobenius norm.\footnote{Since $\|\bA\|$ corresponds to the spectral norm, if $\bA$ is symmetric and positive semi-definite, then we have $\| \bA \| = \lambda_{\max} \leq \sum_{i=1}^n \lambda_i = \| \bA \|_{\bF}$, where $\lambda_{\max} \geq 0$ is the largest eigenvalue of $\bA$, and $\lambda_i \geq 0$ is the $i$th largest eigenvalue of $\bA$. Equality holds if and only if $\bA$ is a rank-1 matrix.} As long as there is no confusion, a vector can be treated as a type of matrix. For instance, using the definition of transpose, the outer product of two vectors $\ba$ and $\bb$ is represented as $\ba \otimes \bb = \ba \bb^\trans$. 

\vskip+1em

\noindent \textbf{Network, uncertainty, and intervention.} Let $N = \{ 1,2, \ldots, n \}$ represent the set of agents. The relationships and interactions among these agents are encapsulated by an $n \times n$ matrix $\bG \in \R^{n \times n}$, referred to as the \textit{influence network}. For a given allocation vector $\bx \in \R^n$, the resulting outcome vector is computed as $\bG \bx$. Each element $\bG_{ij}$ in the influence network characterizes the influence of agent $j$'s allocation on agent $i$'s outcome: a positive $\bG_{ij}$ indicates a beneficial effect, while a negative $\bG_{ij}$ suggests a detrimental one. Thus, $\bG_{ij} x_j$ quantifies the degree to which agent $j$’s allocation influences agent $i$'s outcome. Our analysis does not rely on any specific symmetry in $\bG$. Depending on the intervention context, $\bG$ may be assumed to be symmetric, as in network games \citep[e.g.,][]{Galeottietal:2020:ECMA}, or asymmetric, as in social learning on networks \citep[e.g.,][]{JeongShin2024}.

A DM is tasked with allocating limited resources to agents by selecting the allocation vector $\bx$. While the specific problem faced by the DM will be defined later, we first address the \textit{uncertainty} the DM encounters. The influence network $\bG$ is not fully known to the DM, as it is modeled as a random matrix. It is postulated that the elements of $\bG$, namely $\bG_{ij}$, are \textit{correlated random variables}. The DM is informed of the mean $\bm_{ij}$ and the variance $\bv_{ij}^2$ with $\bv_{ij} > 0$ for each $\bG_{ij}$ but lacks information about the covariance between them for all $i, j \in N$.

\vskip+1em

\noindent \textbf{Adversarial Nature and sequence of decisions.} As in the standard robust optimization problem, the decision-making framework can be illustrated by a sequence of staged decisions, encapsulating interactions with an adversarial entity termed Nature.\footnote{Another interpretation is the existence of multiple priors due to the limited information regarding the network statistics. For a detailed discussion and applications of robust intervention problems, see \cite{Ben-Tal:2009:Book} and references therein.} In the first stage, the DM selects the allocation vector $\bx$ to make $\bG \bx$ as close as possible to the target outcome $\bz$, while accounting for the direct cost of deviating from the reference allocation vector $\bx^0$. 

Subsequently, in the second stage, an adversarial player, \textit{Nature}, ``intervenes'' by choosing the influence network $\bG$ among the agents. Nature's primary objective is to act in direct opposition to the DM's interests, by manipulating the realization of the underlying network among the agents. 

Finally in the third stage, interactions among the agents occur within the established network parameters, ultimately leading to the realization of an allocation outcome. Since the DM acts as the first mover, they must anticipate Nature's potential actions when formulating their strategy. As will be explained in the next section, Nature's action set, called the uncertainty set, is compact due to the known mean and variance of $\bG_{ij}$ for $i,j \in N$.


Therefore, the DM's robust optimization problem is formally defined as follows:
\begin{align}\label{eqn:problem0}
\begin{aligned}
\min_{\bx \in \R^n} \,\max_{\{ {\rm Cov}(\bG_i \bG_j) \}_{i,j \in N}} &\quad \frac{1}{2} \left( \E \left[ \| \bG \bx - \bz \|^2 \right] + \| \bC^{\frac{1}{2}} (\bx - \bx^0) \|^2 \right)  \\
\text{subject to} &\quad \E \left[ \bG_{ij} \right] = \bm_{ij} \quad \text{and} \quad \Var[\bG_{ij}] = \bv_{ij}^2 \quad \text{for all $i,j \in N$},
\end{aligned}
\end{align}
where vector $\bx = (x_1, \dots , x_n)^\trans \in \R^n$ represents the DM's choice of allocation vector across the agents, $\bz = (z_1, \dots , z_n)^\trans\in \R^n$ is the target outcome vector, $\bx^0$ is a reference allocation vector, and $\bC$ is a symmetric positive semi-definite matrix, without loss of generality.\footnote{Without loss of generality, $\bC$ is assumed to be symmetric because the DM's objective function is quadratic. The mathematical expectation in the objective function is multiplied by $\frac{1}{2}$, following the conventional normalization for quadratic objective functions \citep[e.g.,][]{Galeottietal:2020:ECMA, Galeottietal:2024b:WP}.} The first term $\mathbb{E} \left[ \| \mathbf{G} \bx - \mathbf{z} \|^2 \right]$ measures the cost of having an outcome that deviates from the target vector $\mathbf{z}$. The second term $\| \bC^{\frac{1}{2}} (\bx - \bx^0) \|^2$ represents the cost of choosing an allocation that deviates from the reference allocation vector $\bx^0$. The equations specified under the constraints in the optimization problem represent the statistical information---the mean, $ \E \left[ \bG_{ij} \right]$, and variance, $\Var[\bG_{ij}]$---provided to the DM. A detailed interpretation and the implications of the cost function and information structure are collated in Section \ref{subsec:model_implications}.




\section{Analysis}\label{sec:analysis}

\subsection{Decomposition and Assumptions}

\noindent \textbf{Decomposition.} To isolate the distinct aspects of the objective function and derive theoretical insights more efficiently, we decompose it into its component parts. For each pair of agents $i$ and $j$, $\bU_{ij} = \bG_{ij} - \bm_{ij}$ represents the deviation of the influence from its mean. By construction, the mean of $\bU_{ij}$ is zero, and its variance is $\bv_{ij}^2$.\footnote{Formally, $\E[ \bU_{ij}] = 0$ and $\Var[\bU_{ij}] = \E[\bU_{ij}^2] = \Var[\bG_{ij}] = \bv_{ij}^2$ for all $i, j \in N$.} Let $\bm_i = \E [ \bG_i^\trans]$ be the vector of mean influences toward agent $i$, and let $\bU_i = \bG_i^\trans - \bm_i$ represent the vector of deviations of the influences toward agent $i$. Then, the expected value of the squared distance between the allocation outcome to agent $i$ and the target allocation is
\begin{align}\label{eqn:expectationofagenti}
\E[ | \bG_i \bx - z_i |^2 ] &= \sum_{j=1}^n \sum_{k=1}^n x_j \left( \bm_{ij} \bm_{ik} + \E [\bU_{ij} \bU_{ik}] \right) x_k - 2 z_i \sum_{l=1}^n \bm_{il} x_l + z_i^2  \nn  \\
&= \langle \bx, \bM_i \bx \rangle + \langle \bx, \bB_i \bx \rangle - 2 \langle \psi_i, \bx \rangle + z_i^2,
\end{align}
where $\bM_i = \bm_i \otimes \bm_i$, $\bB_i = \E[ \bU_i^\trans \bU_i]$ is the covariance matrix of the influence deviations from the means toward agent $i$, and $\psi_i = z_i \bm_i$ encapsulates the $z_i$-weighted mean influences of the other agents on agent $i$'s outcome. Since $\E[ \| \bG \bx - \bz \|^2 ] = \sum_{i=1}^n \E[ | \bG_i \bx - z_i |^2 ]$, by summing up expression \eqref{eqn:expectationofagenti}, we obtain
\begin{align}\label{eqn:plannercost1}
\E[ \| \bG \bx - \bz \|^2 ] &= \langle \bx, \bM \bx \rangle + \langle \bx, \bB \bx \rangle - 2 \langle \psi , \bx \rangle + \|\bz\|^2,
\end{align}
where $\bM = \sum_{i=1}^n \bM_i$ represents the \textit{aggregated mean network influence}, reflecting the cumulative average effects across all agents. Similarly, $\bB = \sum_{i=1}^n \bB_i$ represents the \textit{aggregated uncertainty of network influence}, encapsulating the collective variability in agent interactions. Finally, $\psi = \sum_{i=1}^n \psi_i = \sum_{i=1}^n z_i \bm_i$ is the $\bz$-weighted sum of mean network influence across the agents, effectively quantifying the targeted influence alignment with the desired outcomes.

Equation~\eqref{eqn:plannercost1} decomposes the effects of changing allocation vector $\bx$ into three distinct channels. The first term, $\langle \bx, \bM \bx \rangle$, encapsulates the effect of $\bx$ through the mean of the network influence structure to the objective function, highlighting how $\bx$ contributes to the predictable aspects of agent interactions. The second term, $\langle \bx, \bB \bx \rangle$ quantifies the impact arising from the inherent uncertainty within the network influence structure, addressing the variability in agent responses that cannot be precisely predicted. The final term, $2 \langle \psi , \bx \rangle$, reflects the direct effect of $\bx$ in aligning the actual outcomes more closely with the target outcomes, emphasizing the strategic alignment of resources. Among these components, the second term uniquely introduces uncertainty through the network. Consequently, the DM must carefully account for all possible realizations of $\bB$, which are determined by an uncertainty set. This uncertainty set, formally defined as $\cB$ in the following paragraph, encapsulates Nature's choices and determines the scope of adversarial scenarios the DM must address.

\vskip+1em

\noindent \textbf{Uncertainty set.} For each agent $i \in N$, let $\cB_i$ denote a set of matrices $\bB_i$ that are symmetric, positive semi-definite,\footnote{This stems from the setup $\bB_i = \E[ \bU_i^\trans \bU_i]$.}  and have the diagonal elements $(\bB_i)_{jj} = \bv_{ij}^2$ for all $j \in N$. This set $\cB_i$, referred to as the \textit{individual uncertainty set}, forms a convex, compact subset within the space of all $n \times n$ real matrices.\footnote{Using the Frobenius norm notation, it also follows that $\| \bB_i \|_{\bF} = \sum_{j=1}^n \bv_{ij}^2$ for all $\bB_i \in \mathcal{B}_i$. Note that for each $i \in N$, $\cB_i$ is contained in a finite-dimensional Euclidean space. Furthermore, each $\cB_i$ is also bounded as $| (\bB_i)_{jk} |  \leq \bv_{ij} \bv_{ik}$ for all $j,k \in N$ by the Cauchy–Schwarz inequality. This condition implies that $| {\rm Cov}(\bG_{ij}, \bG_{ik}) | \leq \bv_{ij} \bv_{ik}$ for all $j,k \in N$.  It is worth mentioning that this boundedness does not imply positive semi-definiteness. Verifying convexity and closedness is straightforward.} The aggregate of these sets is the Minkowski sum of the individual uncertainty sets, defined by $\cB = \{ \bB \in \R^{n \times n} \, \vert \, \bB = \sum_{i=1}^n \bB_i \, \text{for some $\bB_i \in \cB_i$} \}$. This set represents all possible aggregated uncertainties that can be chosen by Nature and is referred to as the ``uncertainty set'' in the literature on robust optimization \citep{Ben-Tal:2009:Book}.\footnote{
Our consideration of the uncertainty set in the network economics context differs from that in the robust mechanism design literature. For instance, in \citet{BirgemanandMorris:ECMA:2005}, the uncertainty set pertains to agents' type spaces, encapsulating all possible payoff types and beliefs about others' types. In \citet{Carroll:AER:2015}, it relates to the agent's technology set, including possible actions represented as output distributions and costs. In \citet{HeandLi:JET:2022}, it arises from the auctioneer's limited information about the correlation structure of bidders' valuations.} 

\vskip+1em

\noindent \textbf{Assumptions.} In order to enhance the clarity of our analysis and ensure the uniqueness of the DM's robust optimization problem \eqref{eqn:problem0}, we impose two key properties throughout the paper:
\begin{itemize}
    \item[(i)] \textbf{\textsf{Property A.}} The aggregated mean network influence $\bM$ has full rank (independent responsiveness condition).
    \item[(ii)] \textbf{\textsf{Property B.}} The solution to the DM's robust optimization problem \eqref{eqn:problem0} contains no zero entries (non-negligence condition). 
\end{itemize}

The two properties have the following economic interpretations. First, note that \textsf{Property A} (independent responsiveness condition) holds if and only if $\{ \bm_i \}_{i=1}^n$ is linearly independent.\footnote{Mathematically, the following are equivalent: (i) $\bM = \sum_{i=1}^n \bm_i \otimes \bm_i$ has full rank, (ii) $\{ \bm_i \}_{i=1}^n$ is linearly independent, and (iii) $\bM$ is a positive definite matrix.} Since $\bm_i$ represents the mean influence toward agent $i$, the full rank property implies that the mean influence to any agent in the network is not a linear combination of the mean influences toward other agents. \textsf{Property A} holds generically; if $\bM$ does not satisfy \textsf{Property A}, then an arbitrarily small perturbation to $\bM$ will be sufficient to make it hold. This property ensures that the DM's robust intervention problem \eqref{eqn:problem0} has a unique solution. On the other hand, \textsf{Property B} (non-negligence condition) ensures that every agent receives a nonzero allocation under the DM’s optimal intervention, preventing any agent from being overlooked. We will show that \textsf{Property B} is sufficient to guarantee a \textit{unique} worst-case scenario, the solution to the inner maximization problem in  \eqref{eqn:problem0} with respect to the DM's optimal intervention.


\textsf{Property B} can be explicitly stated in terms of the model primitives, the function $\displaystyle g(\bx) = \max_{\bB \in \calB} \tfrac{1}{2} \langle \bx, \bB \bx \rangle$, and the set $Z = \{ (\bM + \bC)\bx + \partial g (\bx) \, \vert \, \bx \text{ has a zero entry} \} \subsetneq \R^n$, where $\partial g (\bx)$ is the set of subgradients of $g$ at $\bx$. Specifically, it turns out that \textsf{Property B} holds if and only if $\psi^0 + \psi \in \R^n \setminus Z$, where $\psi^0 = \bC \bx^0$.\footnote{See Online Appendix for proof. In the Online Appendix, we provide additional sufficient conditions under which \textsf{Property B} holds, along with a graphical example of the two-agent case.} This property is satisfied for most parameter values, particularly when $\bC$ is sufficiently large in terms of its smallest eigenvalue.\footnote{For example, if $\bC = c \bI$ for some $c \geq 0$, it suffices to assume that $c$ is sufficiently large, a condition often assumed in the literature on network intervention \citep[e.g.,][]{Galeottietal:2020:ECMA, JeongShin2024}.} Another instance in which \textsf{Property B} holds is when \(\bM\) is large relative to the variance estimates \(\bv_{ij}^2\). In the context of influence network data, this corresponds to the scenario where the DM’s estimate of \(\bM\) is statistically significant, thereby reducing the relative magnitude of the remaining uncertainty. 

We note that the set $Z$ is defined by a set of linear inequalities, making its complement an open set. The existence of a non-empty interior, along with the uniqueness of the DM’s solution, enables comparative static analysis around the solution.


\subsection{Characterization of Unique Robust Intervention}

By employing the decomposition and assumptions from the previous subsection, we fully characterize the unique robust intervention of the DM in the presence of uncertainty in the network structure.

\vskip+1em

\noindent \textbf{Existence of unique robust intervention.} We rewrite the DM's objective function in terms of $\bx$ and $\bB$ using a function $f: \R^n \times \cB \rightarrow \R$ defined as:
\begin{align}\label{eqn:f(x,B)_main}
f(\bx, \bB) = \frac{1}{2} \big( \langle \bx, \bM \bx \rangle + \langle \bx, \bB \bx \rangle + \langle \bx, \bC \bx \rangle - 2 \langle \psi^0 + \psi, \bx \rangle + \underbrace{\| \bz \|^2 + \| \bC^{\frac{1}{2}} \bx^0 \|^2}_{\text{constant}} \big), 
\end{align} 
where $\| \bC^{\frac{1}{2}} (\bx - \bx^0)\|^2 = \langle \bx, \bC \bx \rangle - 2 \langle \psi^0, \bx \rangle + \| \bC^{\frac{1}{2}} \bx^0 \|^2$ and $\psi^0 = \bC \bx^0$. The function $f(\bx, \bB)$ is strictly convex in $\bx$ because the matrix $\bM$ is positive definite by \textsf{Property A}, while the matrices $\bB$ and $\bC$ are positive semi-definite. Furthermore, $f(\bx, \bB)$ is linear in each individual uncertainty $\bB_i$ and hence in the aggregated uncertainty $\bB$. These characteristics facilitate the formulation of the optimization problems faced by the DM and the adversarial Nature. 

Utilizing expression \eqref{eqn:f(x,B)_main}, we can formulate the DM's robust optimization problem \eqref{eqn:problem0} and the corresponding dual problem for the adversarial Nature, structured as dual robust optimization problems:
\vskip+1em
\begin{minipage}[ht]{0.9\textwidth}
\centering
\begin{minipage}[ht]{.4\textwidth}
\centering{\textbf{DM's primal problem}} 
\begin{align}\label{eqn:DM}\tag{(PP)}
\begin{aligned}
\min_{\bx \in \R^n} \max_{\bB \in \cB} \ &f(\bx, \bB). 
\end{aligned}
\end{align}
\end{minipage}
\hspace{0.2 in}
\begin{minipage}[ht]{.4\textwidth}
\centering{\textbf{Nature's dual problem}} 
\begin{align}\label{eqn:Nature}\tag{(DP)}
\begin{aligned}
\max_{\bB \in \cB} \min_{\bx \in \R^n} \ &f(\bx, \bB). 
\end{aligned}
\end{align}
\end{minipage}
\end{minipage}

\vskip+1em

Equipped with the duality between the two problems \citep{Neumann:1928:MathAnn}, we can find a unique solution to the DM's primal problem \eqref{eqn:DM}, in relation to Nature's dual problem \eqref{eqn:Nature}:


\begin{theorem}[Duality and Uniqueness]\label{theorem:duality}  
The following duality holds:
\begin{align}\label{duality}
\min_{\bx \in \R^n} \, \max_{\bB \in \cB} f(\bx, \bB) = \max_{\bB \in \cB}  \, \min_{\bx \in \R^n} f(\bx, \bB).
\end{align}
Furthermore, \eqref{eqn:DM} and \eqref{eqn:Nature} each have a unique solution, denoted by $\xop$ and $\bB^*$, which are related through the equation $(\bM + \bB^* + \bC) \bx^* = \psi^0 + \psi$.
\end{theorem}

Given DM's intervention strategy $\bx$, let $\bB_{BR}(\bx)$ denote Nature's aggregated best response defined as a solution to the inner optimization problem $\displaystyle \max_{\bB \in \cB} f(\bx, \bB)$ in \eqref{eqn:DM}. Similarly, let $\bx_{BR}(\bB)$ denote DM's best response given Nature's choice $\bB$. The first part of \autoref{theorem:duality} establishes that the optimal values of both the original and the dual problems are equal; that is, $f(\xop,\bB_{BR}(\xop)) = f(\bx_{BR}(\bB^*),\bB^*)$.\footnote{This best-response relationship provides another characterization of the DM's solution as an action profile of a game between the DM and Nature, in which DM tries to minimize $f(\bx, \bB)$ but Nature tries to maximize $f(\bx, \bB)$ \citep[e.g.,][]{Neumann:Book:2007}. Our proof in \autoref{sec:appendix:proofs} builds on this interpretation.} 

However, the duality \eqref{duality} does not necessarily imply that Nature's best response $\bB_{BR}(\xop)$ with respect to $\xop$ directly solve Nature's dual problem \eqref{eqn:Nature}. In other words, solving the dual problem \eqref{eqn:Nature} using the backward induction does not necessarily resolve the primal problem \eqref{eqn:DM}. The second part of \autoref{theorem:duality} shows that under \textsf{Property A} and \textsf{Property B}, both \eqref{eqn:DM} and \eqref{eqn:Nature} have unique solutions that are mutually best responses. This result is expressed by the equation $(\bM + \bB^* + \bC) \bx^* = \psi^0 + \psi$.

The equation $(\bM + \bB^* + \bC) \bx^* = \psi^0 + \psi$ offers an interpretation of the solution in terms of the model parameters. First, $\xop$ is adjusted to align with the mean influences represented by $\bM$. Second, $\xop$ takes into account the risk by responding optimally to Nature's adversarial choice $\bB^*$. Third, $\xop$ remains cost-effective, considering the cost structure given by $\bC$. The combined influence of the initial allocation $\psi^0$ and the target allocation $\psi$ guides $\xop$ towards a balance between the current state of intervention and its desired outcome.



\vskip+1em

\noindent \textbf{Characterization of the unique worst-case scenario.} We now characterize a unique worst-case scenario as Nature's optimal choice and determine the corresponding best response of the DM. In the DM's objective function \eqref{eqn:f(x,B)_main}, the only term that depends on Nature's choice $\bB$ is $\langle \bx, \bB \bx \rangle$. Thus, in the dual problem \eqref{eqn:Nature}, Nature tries to maximize $\langle \bx_{BR}(\bB), \bB\, \bx_{BR}(\bB) \rangle$. In fact, Nature can equivalently solve \eqref{eqn:Nature} by maximizing $\langle \xop, \bB \xop \rangle$, that is, $\bB^*$ is characterized as $\bB^* = \bB_{BR}(\xop)$,  as shown in \autoref{sec:appendix:proofs}. Consequently, the properties of $\bB^*$ depend on the properties of $\xop$. 

Indeed, if $\xop$ does not contain any zero entry by \textsf{Property B}, it turns out that there exists a unique worst-case scenario (i.e., Nature's most adversarial choice) for the DM's choice. To understand why, consider an example with two agents in which Nature maximizes $\langle \xop, \bB \xop \rangle = \langle \xop, \bB_1 \xop \rangle + \langle \xop, \bB_2 \xop \rangle$. The first term expresses the uncertainty generated by influences toward agent 1's outcome, while the second term represents the uncertainty generated by influences toward agent 2's outcome. We focus on Nature's optimal choice of $\bB_1$, as the choice of $\bB_2$ can be similarly understood.

Note that $\bB_1$ is a symmetric positive semi-definite matrix, with its diagonal entries being $\bv_{11}^2$ and $\bv_{12}^2$, and its off-diagonal entry $\rho_1 = {\rm Cov}(\bG_{11}, \bG_{12})$ is undetermined, bounded by $| \rho_1 | \leq \bv_{11} \bv_{12}$. The term $\langle \bx^*, \bB_1 \bx^* \rangle$ reaches its maximum at the extreme value of $\rho_1$; consequently, $\rho_1 = \bv_{11} \bv_{12} \text{ or } -\bv_{11} \bv_{12}$. In addition, by the spectral theorem \citep{Meyer:Book:2010}, there are two non-negative eigenvalues, $\lambda_1$ and $\lambda_2$, and corresponding unit-length eigenvectors, $\bw_1$ and $\bw_2$, such that $\bB_1 = \lambda_1 (\bw_1 \otimes \bw_1) + \lambda_2 (\bw_2 \otimes \bw_2)$. Interestingly, when $\rho_1 = \bv_{11} \bv_{12} \text{ or } -\bv_{11} \bv_{12}$, the first eigenvalue is $\| \bB_1 \|_{\bF}^2 = \bv_{11}^2 + \bv_{12}^2$, which is constant for all $\bB_1 \in \mathcal{B}_1$. Therefore, $\bB_{BR}(\xop) = (\bv_{11}^2 + \bv_{12}^2) (\bw_1 \otimes \bw_1)$. 

To determine the sign of $\rho_1$, consider the entries of $\xop$. Since $\xop$ has no zero entry, two cases arise: (1) the entries of $\xop$ have the same signs, or (2) the entries have opposite signs. In the first case, the covariance of $\bG_{11}$ and $\bG_{12}$ must be positive, as a negative covariance would create a hedging effect, which is suboptimal for Nature. Similarly,  in the second case, where the signs are opposite, the covariance of $\bG_{11}$ and $\bG_{12}$ must be negative. Thus, $\rho_{1} = s(\bx_1^*) s(\bx_2^*) \bv_{11} \bv_{12}$, which leads to $\bw_1 = \frac{\bq_1}{\|\bq_1\|}$, where $\bq_1 = ( s(\bx_1^*) \bv_{11}, s(\bx_2^*) \bv_{12})^\trans$ and is located in the same orthant with $\xop$. 

The above analysis shows that Nature's optimal strategy concentrates uncertainty in the first principal component, aligned with the DM's intervention vector (in the sense that they are in the same orthant). By doing so, Nature maximizes the uncertainty effect relative to the DM's choice. Since the analysis for agent 1 is independent of that for agent 2, we obtain a unique closed-form expression for $\bB_i^* = \sigma_i^2 (\bw_i \otimes \bw_i)$, where $\sigma_i^2 = \sum_{j=1}^2 \bv_{ij}^2$ and $\bw_i = \frac{\bq_i}{\|\bq_i\|}$ with $\bq_i = (s(x_1^*) \bv_{i1}, s(x_2^*) \bv_{i2})^\trans$ for $i = 1,2$.

The analysis is closely related to the orthogonal decomposition techniques commonly used in the literature on intervention in networks \citep[e.g.,][]{Galeottietal:2020:ECMA, Galeottietal:2024b:WP, JeongShin2024}. In those models, an orthogonal decomposition tracks the effects of the DM's choice across multiple dimensions, with the network given as a model parameter, so the principal components of the underlying network are fixed. However, in the current model, for each agent $i$, the unique principal component with a non-zero eigenvalue, $\bq_i$, is optimally chosen by Nature in the dual problem, with an objective function opposing that of the DM. Consequently, the adversarial Nature focuses on a particular principal component and concentrates all its effects on each agent $i$ through this choice of principal component. 

\begin{theorem}[Unique Worst-Case Scenario]\label{thm:rank-1}
There exists a unique solution to Nature's dual problem \eqref{eqn:Nature}, $\bB^* = \sum_{i=1}^n \bB_i^*$ in which $\bB_i^* = \sigma_i^2 (\bw_i \otimes \bw_i)$ with $\sigma_i^2 = \sum_{j=1}^n \bv_{ij}^2$, $\bw_i = \frac{\bq_i}{\|\bq_i\|}$, and $\bq_i = (s(\bx_1^*) \bv_{i1}, \dots, s(\bx_n^*) \bv_{in})^\trans$ for all $i \in N$. Consequently, the unique solution to problem \eqref{eqn:DM} is expressed by the following equation:
\begin{align}\label{eqn:xop_sol}
\xop = \left( \bM + \sum_{i=1}^n \sigma_i^2 \left( \frac{\bq_i}{\|\bq_i\|} \otimes \frac{\bq_i}{\|\bq_i\|} \right) + \bC \right)^{-1} (\psi^0 + \psi).    
\end{align}
\end{theorem}

\autoref{thm:rank-1} provides an algorithmic approach toward finding the robust solution by applying the principle of backward induction. First, consider the interiors of all orthants in $\R^n$. For example, if $n=2$, there are four open quadrants to consider. Second, since Nature's optimal choice in the dual problem \eqref{eqn:Nature} depends only on the signs of the entries of $\xop$, determine the possible values of Nature's choice for each orthant. For example, if $n=2$, identify four possible values of $\bq_i = (s_{i1} \bv_{i1}, s_{i2} \bv_{i2})^\trans$, $i=1,2$, corresponding to each orthant represented by $(s_{i1}, s_{i2}) \in \{ +1, -1 \}^2$. Then calculate the corresponding $\xop$ in expression \eqref{eqn:xop_sol} and check whether $\xop$ remains in the same orthant. If so, that $\xop$ is the unique robust intervention solution. Otherwise, try another orthant and repeat the process, which will end in a finite number of steps.

We conclude by noting that the random influence matrix $\bG$ can be assumed symmetric in \autoref{thm:rank-1}. One might think that constructing the covariance matrix $\bB_i$ for the links toward agent $i$ could influence the construction of $\bB_j$ for the links toward another agent $j$, and that due to the symmetry of $\bG$, $\bB_j$ might not be a rank-1 matrix. However, this is not the case, as shown in the proof of \autoref{thm:rank-1}. 

\section{Comparative Analysis and Model Implications}

\subsection{Cost of Uncertainty}\label{sec:definitions}

We now investigate how changes in the degree of uncertainty affect the performance of the robust intervention. In our model, uncertainty is introduced by an adversarial Nature under specific constraints. Since the magnitude of the uncertainty constraint is bounded by the variance of the entries in the influence matrix $\bG$, we examine how changes in these variances influence the performance of the robust intervention---a concept we refer to as the \textit{cost of uncertainty}.

There are two possible measures of the cost of uncertainty: (i) global uncertainty and (ii) local uncertainty. The cost of global uncertainty reflects the system-wide impact of uncertainty across the entire network. The DM's uncertainty about the network structure is captured by the aggregated uncertainty matrix $\bB^* = \sum_{i=1}^n \bB_i^*$, which represents the worst-case scenario chosen by Nature. This aggregation means that $\bB^*$ encompasses the combined effect of all uncertainty sources rather than a single localized variation. Since $\bB^*$ is positioned at the boundary of Nature's uncertainty set $\calB$, any expansion of this set alters the DM’s objective function. Thus, by applying the envelope theorem, we define the cost of global uncertainty as the partial derivative of the DM's objective function with respect to $\bB^*$, evaluated at $\bx^*$.

Regarding the cost of local uncertainty, we observe that at the boundary of the uncertainty set, the level of correlation is maximized to match the level of variance of the associated component, measured by $\| \bB_i \|_\textbf{F}$ for each agent $i$. Since $\| \bB_i \|_\textbf{F}$ represents the sum of the variances of the links toward agent $i$, the relevant aggregated uncertainty level for the DM is the sum of these variances. Consequently, we can measure the cost of local uncerntainty as the partial derivative of the DM's objective function with respect to $\bv_{ij}^2$, evaluated at the robust intervention $\bx^*$. This approach measures the change in the objective function due to the variance of a particular link in the network, while the cost of global uncertainty captures the impact of simultaneous changes in all the variances. 

\vskip+1em

\noindent \textbf{Cost of global uncertainty.} When the objective function is perturbed by a change in $\bB$ in the direction of $\Delta \bB$, the directional change is given by ${\rm Trace} \left( \frac{\partial f(\bx_{BR}(\bB), \bB)}{\partial \bB}^\trans \Delta \bB \right)$. Using standard matrix calculus, we find that $\frac{\partial f(\bx_{BR}(\bB), \bB)}{\partial \bB}\big|_{\bB = \bB^*} = \xop \otimes \xop$ by the envelope theorem, which is positive semi-definite and rank-1. 
This indicates that perturbations in the function $f(\bx_{BR}(\bB), \bB)$ are highly directional, with the function being most sensitive to changes in $\bB$ that align with the direction of $\bx$. Consequently, at the robust intervention solution $\bx_{BR}(\bB^*) = \xop$, we observe that the directional change is maximized when $\Delta \bB$ aligns with the direction of the robust intervention $\bx^*$.


\begin{proposition}[Cost of Global  Uncertainty]\label{prp:vofi}
The cost of global uncertainty is $\xop \otimes \xop$, which is positive semi-definite and rank-1.\footnote{\textsf{Property B} is necessary to ensure that the derivative is well-defined.}
\end{proposition}

The proposition suggests that the economic impact of global uncertainty is entirely concentrated in the direction of robust intervention, $\xop$. This means that the DM faces the highest risk precisely where they have allocated resources, making the intervention particularly vulnerable to uncertainty in this specific direction. Since the cost of global uncertainty is rank-1, Nature's worst-case response strategically amplifies risk only along the chosen intervention strategy, rather than spreading it uniformly across all agents. Consequently, even a small increase in uncertainty along $\xop$ can disproportionately affect outcomes, compelling the DM to internalize this risk when designing robust interventions. This highlights a fundamental trade-off: optimizing for mean influence versus hedging against uncertainty-driven disruptions.
 

\vskip+1em

\noindent \textbf{Cost of local uncertainty.} We now investigate the cost of local uncertainty. First, note that $(\bB_i^*)_{jk} = \bv_{ij} \bv_{ik} s(x_j^*) s(x_k^*)$ for any $i,j,k \in N$, which represents the covariance of agent $j$'s and agent $k$'s allocations toward agent $i$'s outcome at the worst-case scenario determined by Nature. Thus, when $\bv_{ij}$ changes for some $i,j \in N$, it affects not only the variance of the directional link from agent $j$ to agent $i$, but also all the covariances that are associated with the two agents. Specifically, recall that adversarial Nature's choice maximizes the covariance of the agents within the constraint of uncertainty. For a given agent $i$, an increase in the variance of the link from agent $j$'s allocation to agent $i$'s outcome raises the upper bound of the covariance between this link and all other links from the other agents to $i$, including the link from $j$ to $i$. Consequently, the partial derivative of the DM's objective function with respect to $\bv_{ij}$ is given by $2 \sum_{k=1}^n \bv_{ik} |x_j^*| |x_k^*|$, which is strictly positive if \textsf{Property B} holds, ensuring the unique existence of the derivative, as with the cost of global uncertainty. 

Unlike the cost of global uncertainty, where a change in $\bB$ is multi-directional, the cost of local uncertainty is strictly positive due to the unidirectional change of $\bv_{ij}$ (i.e., only in an increasing direction). 

\begin{proposition}[Cost of Local Uncertainty]\label{prp:vofli}
The cost of local uncertainty for the link from agent $j$ to agent $i$ is $2 \sum_{k=1}^n \bv_{ik} \, |x_j^*| |x_k^*|$.
\end{proposition}

The proposition suggests that local uncertainty affects not only a single link but also the entire network through covariance interactions. Unlike global uncertainty, which perturbs all variances in a structured way, local uncertainty increases a specific link’s variance while amplifying its covariances with other links. This means that uncertainty at one agent influences both its own effectiveness and the overall robustness of the intervention. At the robust solution $\xop$, the cost of local uncertainty is strictly positive, as any increase in variance propagates worst-case risk across multiple connections. Since Nature maximizes these covariances within the uncertainty constraint, the DM must weigh the trade-off between targeting high-influence agents and mitigating localized uncertainty that could destabilize the intervention strategy.

\subsection{Balancing Uncertainty and Influence: Two-Agent Model}
\label{subsec:roleofuncertainty}


In this subsection, we examine a two-agent network example to provide concrete insight into the two costs of uncertainty introduced in Section~\ref{sec:definitions}. This example illustrates how the DM balances minimizing uncertainty costs and targeting an agent with higher mean influence. It also demonstrates how global and local uncertainty shape the optimal intervention strategy by determining when uncertainty at one agent outweighs the benefits of its mean influence. In particular, the analysis introduces an uncertainty threshold, showing that once uncertainty exceeds this level, reallocating resources to the less uncertain agent with lower mean influence becomes optimal. Through this example, we explore how uncertainty propagates through the network and how the DM adjusts intervention accordingly.

Consider the network of two agents illustrated in \autoref{fig:example_two_nodes1}. For simplicity, we parameterize the network such that $\bm_{11} = \bm_{21} = m > \bm_{22} = \bm_{12} = 1$, so that any allocation to agent 1 generates a strictly greater mean externality than one to agent 2. To capture the uncertainty in the agents' influences, we set $\bv_{11} = \bv_{21} = v > \bv_{22} = \bv_{12} = 1$, implying that an allocation to agent 1 generates higher uncertainty compared to agent 2. The two agents are otherwise identical, with $\bx^0 = \bf{0}$, $\bz = (1,1)^\trans$, and $\bC = c \bI$ for some sufficiently large $c > 0$. These parameters allow us to explore the trade-off between agent 1's higher mean externality and greater uncertainty.

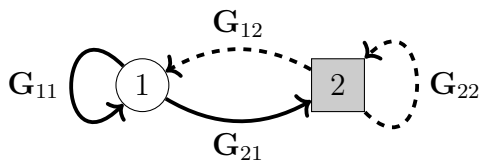
\begin{figure}[ht]
\centering
\vskip-1em
\begin{tikzpicture}
[scale=0.65, state/.style={circle, fill = white, draw=black!80, inner sep = 0.02 in, line width = 0.015 in}]
\node[circle, draw=black, inner sep=0pt, minimum size = 20 pt] (n1) at (0,0)  {1};
\node[rectangle, draw=black, inner sep=0pt, minimum size = 20 pt, fill=black!20] (n2) at (4,0)  {2};
\path[line width = 0.02 in, ->] (n1) edge [out=135,in=-135,distance=2 cm] node [anchor = east]{$\bG_{11}$} (n1);
\path[line width = 0.02 in, ->, dashed] (n2) edge [bend right] node [anchor = south]{$\bG_{12}$} (n1);
\path[line width = 0.02 in, ->, dashed] (n2) edge [out=-45,in=45,distance=2 cm] node [anchor = west]{$\bG_{22}$} (n2);
\path[line width = 0.02 in, ->] (n1) edge [bend right] node [anchor = north]{$\bG_{21}$} (n2);
\end{tikzpicture}
\vskip-1em
\caption{Example of two-agent network. Each arrow represents the degree of influence of the allocation given to the starting agent to the heading agent. Since $\bm_{11} = \bm_{21} > \bm_{12} = \bm_{22}$, thick arrows represent a higher mean influence than dashed arrows.}\label{fig:example_two_nodes1}
\end{figure} 

We calculate the aggregated mean influence and uncertainty matrices $\bM$ and $\bB$:
\begin{align*}
\bM =  2
\begin{bmatrix}
m^2 & m \\
m & 1 
\end{bmatrix}
\quad \text{and} \quad \bB 
= 2
\begin{bmatrix}
v^2              & \frac{\rho_1 + \rho_2}{2} \\
\frac{\rho_1 + \rho_2}{2}  &      1
\end{bmatrix},
\end{align*}
where $\rho_1 \in [-v,v]$ is the covariance of $\bG_{11}$ and $\bG_{12}$, and $\rho_2 \in [-v, v]$ is the covariance of $\bG_{21}$ and $\bG_{22}$. \textsf{Property A} is satisfied because $m > 1$. \textsf{Property B} holds as $(\psi^0 + \psi) = 2(m,1)^\trans \notin Z = \{ (\bM + \bC)\bx + \partial g (\bx) \, \vert \, \bx \text{ has a zero entry} \}$ as $c > 0$ is sufficiently large.

A natural guess for the orthant in which $\bx^*$ is contained is the first quadrant. Then, by \autoref{thm:rank-1}, Nature concentrates all the uncertainty in the first principal component, aligned with the DM's intervention vector in the first quadrant as
\begin{align*}
\bB^* 
= 2
\begin{bmatrix}
v^2    & v \\
v      &      1
\end{bmatrix}
= 2 \left( \begin{bmatrix}
v \\ 
1
\end{bmatrix} \otimes
\begin{bmatrix}
v \\ 
1
\end{bmatrix}  \right).
\end{align*}
Consequently, we calculate the optimal robust intervention solution $\bx^*$ as
\begin{align}\label{eqn:xop_example}
\bx^* &\!=\! \left( \bM \!+ \!\bB^* \!+\! c \bI \right)^{-1} (\psi^0 \!+\! \psi) 
\!=\!
\begin{bmatrix}
m^2 \!+\! v^2 \!+\! \frac{c}{2} & m \!+\! v \\
m+v     & 2 + \frac{c}{2} 
\end{bmatrix}^{-1}
\begin{bmatrix}
m \\
1
\end{bmatrix}
\!=\!
\frac{1}{\triangle}
\begin{bmatrix}
(1 \!+\! \frac{c}{2})m \!-\! v \\
v^2 \!+\! \frac{c}{2} \!-\! mv
\end{bmatrix},
\end{align}
where $\triangle > 0$ is the determinant of $\bM + \bB^* + c \bI$. The DM allocates more resources to agent 1 (i.e., $x_1^* \geq x_2^*$) if and only if the uncertainty associated with agent 1 is smaller than a threshold $\overline{v}(m)$ that is calculated as
\begin{align*}
\overline{v}(m) = \frac{1}{2} \left( (m-1) + \sqrt{2 c (m - 1) + (m + 1)^2 } \right).
\end{align*}
The threshold $\overline{v}(m)$ is strictly increasing and concave in $m$.

The above results illustrate the trade-off between leveraging higher mean influence and managing the risks of increased variance. The increasing threshold $\overline{v}(m)$ implies that as agent 1 becomes more influential in terms of individual mean externality, a higher level of uncertainty can still justify allocating more resources to that agent. In other words, the benefit of targeting agent 1 (higher mean influence) grows as $m$ increases, making it more acceptable to tolerate higher uncertainty in the influence of that agent. The concavity of the threshold implies that the rate at which the threshold uncertainty increases diminishes as $m$ grows larger. Initially, small increases in $m$ allow for significant increases in acceptable uncertainty, but further increments in $m$ result in smaller increases in the acceptable level of uncertainty. The following proposition summarizes the result:
\begin{proposition}\label{prp:threshold_v}
There exists a threshold uncertainty level $\overline{v}(m)$ such that $x_1^* \geq x_2^*$ if and only if $v \leq \overline{v}(m)$. $\overline{v}(m)$ is strictly increasing and concave in $m$.
\end{proposition}

The economic insight from \autoref{prp:threshold_v} offers an important extension to the literature on intervention in networks. In models such as \citet{Galeottietal:2020:ECMA} and \citet{JeongShin2024}, the common recommendation is to target agents with high centrality, as their mean influence is crucial for maximizing externalities. However, the current model introduces a new dimension by incorporating uncertainty into the decision-making process. When the DM's objective accounts not only for mean influence but also for the associated risks and robustness, targeting highly central agents may no longer be optimal. This broadens the existing literature by highlighting the trade-off between leveraging high mean externalities and managing uncertainty---an aspect that previous models have often overlooked. 

We close this subsection by linking \autoref{prp:threshold_v} to the two-agent example and its relationship with \autoref{prp:vofi} and \autoref{prp:vofli}. The cost of global uncertainty (\autoref{prp:vofi}) is rank-1, meaning risk is concentrated along the intervention strategy $\xop$. As the variance of agent 1's influence ($v$) increases, it amplifies worst-case risk, making the DM's allocation highly sensitive to uncertainty. The cost of local uncertainty (\autoref{prp:vofli}) further compounds this effect by propagating risk through network interactions, as uncertainty in one agent affects linked agents via covariance. In the two-agent example, these effects reinforce each other: global uncertainty increases direct risk exposure along $\xop$, while local uncertainty amplifies instability by spreading correlated risks across agents. When $v$ surpasses the threshold $\overline{v}(m)$, these combined uncertainty costs outweigh agent 1's higher mean influence, prompting the DM to shift resources to agent 2 for a more stable intervention outcome.

This distinction between global and local uncertainty is crucial for designing robust interventions. While global uncertainty concentrates risk in the direction of the intervention strategy $\bx^*$, local uncertainty propagates through the network, creating correlated disruptions that the DM must mitigate. Accounting for both types of uncertainty ensures a more resilient allocation strategy.




\subsection{Model Implications}
\label{subsec:model_implications}

\noindent \textbf{Uncertainty in the network structure.} The optimized value of the DM's objective function, evaluated at the optimal solution, quantifies the guaranteed outcome of the original objective under the ``worst-case'' scenario chosen by Nature. A key contribution of this paper is the comprehensive characterization of the DM's solution to the robust intervention problem, along with the corresponding worst-case scenarios induced by Nature. This characterization advances the literature by extending previous analyses \citep[e.g.,][]{Belhajetal:GEB:2023, Galeottietal:2020:ECMA, Galeottietal:2024b:WP, JeongShin2024}, which assumes a fully known influence network entry $\bG_{ij}$, to cases where the DM has access only to the first and second moments of $\bG_{ij}$. 

There are several possible interpretations of the influence network $\bG$. For instance, $\bG$ may represent the exchange or redistribution of allocated resources among agents. In the context of food stamps or coupons, citizens might exchange these benefits informally based on their individual needs. Such exchanges, though beneficial for addressing localized shortages, may not be fully observed or known to the distributor. This uncertainty forces the distributor to adopt a robust optimization approach, accounting for potential discrepancies in the network of exchanges. By considering this uncertainty, the distributor can ensure that, even after unobserved exchanges occur, the final distribution of food stamps more closely aligns with a target distribution.

A similar challenge emerges in the context of distributing medical resources during a pandemic, such as facial masks, vaccines, or protective equipment. A social planner such as government or health authority allocates these scarce resources to regions or populations based on observed needs, but individuals or institutions may redistribute them informally or through local agreements. For instance, a region facing an unexpected surge in demand might reallocate vaccines or masks from neighboring areas, leading to deviations from the original distribution plan. Such redistribution, while addressing immediate needs, introduces uncertainty into the overall allocation strategy of the social planner. To account for this, a robust optimization approach becomes essential, allowing the distributor to design an allocation strategy that ensures the final distribution of medical resources, after such redistributions, still aligns with public health objectives, such as achieving herd immunity or maintaining adequate protective equipment coverage across regions.

Another example arises in game-theoretic contexts, such as public goods games on networks \citep[e.g.,][]{Galeottietal:2020:ECMA}. In these games, agents decide their contributions to a public good, with each agent's optimal contribution—or best-response action—depending on their neighbors' contributions, the network structure, and their individual endowments. The DM can influence agents' endowments to align their best-response actions with a desired target profile, such as an efficient or equitable outcome. However, in networks with non-reciprocal relationships, such as unequal mutual benefits or parasitism \citep[e.g.,][]{Bayeretal:JET:2023}, agents' myopic best-response dynamics may fail to converge to a Nash equilibrium. In such cases, the DM's goal is to guide these dynamics toward stability and a target behavior.

Influence network $\bG$ can also be interpreted in the context of viral marketing and social learning \citep[e.g.,][]{JeongShin2024}. In social networks, consumers exchange opinions about products through their interactions, and firms invest substantial resources to influence this opinion formation process. This process is often modeled as a DeGroot learning process \citep[e.g.,][]{DeGroot1974reaching, DeMarzoetal:2003:QJE, GolubJackson2010, GolubandJackson:2012:QJE}, where each agent's opinion is a weighted average of their neighbors' opinions. The DM may intervene in this process by influencing agents' opinions, for instance, through free samples, discounts or targeted advertising. When the DM has limited information about the correlations among consumer connections, they may adopt a robust optimization approach to ensure effective interventions that achieve desired outcomes, even under network uncertainty.

\vskip+1em

\noindent \textbf{Cost structure of the model.} Here, we discuss the cost structure of the model represented by the second term in the DM's robust optimization problem \eqref{eqn:problem0}, $\| \bC^{\frac{1}{2}} (\bx - \bx^0) \|^2$. This term $| \bC^{\frac{1}{2}} (\bx - \bx^0) |^2$ can represent various situations depending on the context. For example, in a budget allocation problem, one could set $\bC = c \bp \bp^\trans$ with some $c > 0$ and $\bp \in \R_{++}^n$ with $\| \bp \|^2 = 1$ for normalization. In this case, the expression simplifies to $\| \bC^{\frac{1}{2}} (\bx - \bx^0) \|^2 = c | \langle \mathbf{p}, \bx - \bx^0 \rangle |^2$. This constraint reflects the cost of choosing an allocation that deviates from a benchmark budget line passing the status quo allocation vector $\bx^0$, as illustrated in \autoref{fig:two_costs}-(a).\footnote{This cost structure represents a scenario where the DM, as a budget allocator, starts with a benchmark budget line passing through the status quo allocation $\bx^0$. Allocations that lie on this benchmark line do not incur additional costs. However, if the allocation deviates from the line in either the northeast direction (spending more than the given budget) or the southwest direction (under utilizing the budget and wasting part of it) in \autoref{fig:two_costs}-(a), the DM incurs a cost proportional to the extent of the deviation, measured by $c |\langle \mathbf{p}, \bx - \bx^0 \rangle |^2$. } 

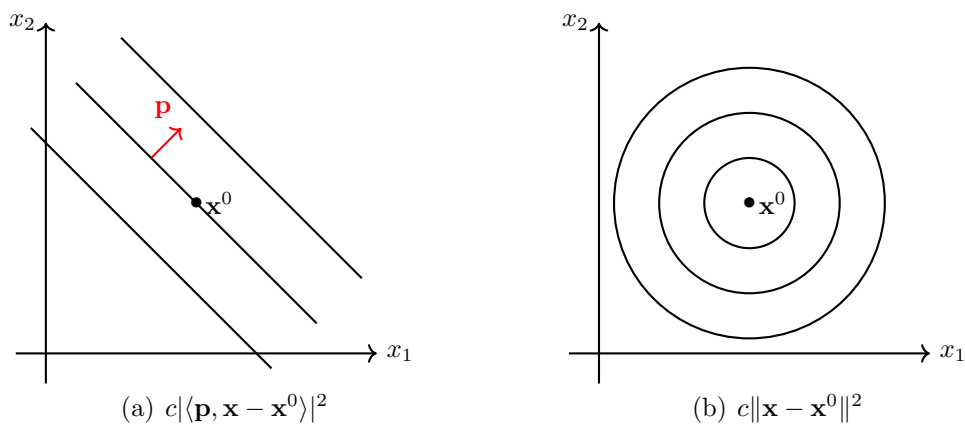
\begin{figure}[ht]
\centering
\footnotesize
\subfigure[$c | \langle \bp, \bx - \bx^0 \rangle |^2$]{
\begin{tikzpicture}[scale = 0.20]

\draw[->, thick] (-12,-10) -- (12,-10) node[anchor = west] {$x_1$};
\draw[->, thick] (-10,-12) -- (-10,12) node[anchor = east] {$x_2$};

\draw[-, thick] (-8,8) -- (8,-8) node[anchor = west] {};
\draw[-, thick] (-5,11) -- (11,-5) node[anchor = west] {};
\draw[-, thick] (-11,5) -- (5,-11) node[anchor = west] {};

\draw (0, 0) node {\textbullet};
\draw (0, 0) node[anchor = west] {$\mathbf{\bx}^0$};

\draw[->, thick, red] (-3, 3) -- (-1, 5) node[anchor = south east] {$\bp$};

\end{tikzpicture}
}
\hspace{0.5 in}
\subfigure[$c \| \bx - \bx^0 \|^2$]{
\begin{tikzpicture}[scale = 0.20]
\draw[->, thick] (-12,-10) -- (12,-10) node[anchor = west] {$x_1$};
\draw[->, thick] (-10,-12) -- (-10,12) node[anchor = east] {$x_2$};

\filldraw[color=black!, fill=none, thick] (0,0) circle (3);
\filldraw[color=black!, fill=none, thick] (0,0) circle (6);
\filldraw[color=black!, fill=none, thick] (0,0) circle (9);

\draw (0, 0) node {\textbullet};
\draw (0, 0) node[anchor = west] {$\mathbf{\bx}^0$};

\end{tikzpicture}
}
\caption{Illustration of the two cost functions}\label{fig:two_costs}
\end{figure}

Another example involves setting $\bC = c \bI$ for some $c > 0$, leading to $\| \bC^{\frac{1}{2}} (\bx - \bx^0) \|^2 = c \| \bx - \bx^0 \|^2$. This represents the cost associated with deviating from the status quo allocation vector $\bx^0$, as illustrated in \autoref{fig:two_costs}-(b). Of course, a combination of these two constraints is possible.\footnote{For example, $c_1 | \langle \mathbf{p}, \bx - \bx^0 \rangle |^2 + c_2 \| \bx - \bx^0 \|^2$ for some $c_1$ and $c_2 > 0$ represents a possible combination of the two constraints of different types.} 


\section{Extensions}
\label{sec:applicationsandextensions}



\subsection{Robust Intervention with Network Expansion}

We now extend the model to account for the introduction of a new agent into an existing network, focusing on how the DM can determine a robust intervention despite the additional uncertainties introduced by this inclusion. Specifically, a new agent joins a network of $n$ existing agents. The DM has full knowledge of the covariance matrix among the existing $n$ agents (i.e., ${\rm Cov}(\bG_{ij}, \bG_{ik})$ for each $i,j,k \in \{1,\dots,n\}$), but has limited or no knowledge of the covariance between the new agent and the existing agents. For instance, consider a scenario where a new researcher joins an existing collaboration network. The DM must decide how to allocate the research budget among the $n+1$ agents. While the DM understands how the $n$ existing members interact with one another, based on research history and outputs, the covariances between the existing members and the new member remain (partially) uncertain.

As an example, consider a network consisting of two existing members (agent 1 and agent 2) and one new agent (agent 3). For simplicity, we assume that the variance of the link is one for all agents; that is, all the diagonal entries of each $\bBbar_i$ are equal to one for $i = 1,2,3$. Consequently, each entry of $\bBbar_i$ matrix represents a correlation coefficient of the influences toward agent $i$. We assume that agent 1's influence on their own outcome is uncorrelated with agent 2's influence on agent 1's outcome (i.e., $\left( \bB_1 \right)_{12} = 0$). In contrast, we assume that agent 1's influence on agent 2's outcome is correlated with agent 2's influence on agent 2's outcome (i.e., $\left( \bB_2 \right)_{12} = \rho \in [-1, 1]$), and the exact correlation coefficient $\rho$ is known to the DM. Finally, we suppose that all the other entries are unknown to the DM. Under these conditions, the following matrices illustrate the network uncertainty:
\begin{align*}
\overline{\bB}_1 = 
\kbordermatrix{
       & 1        & 2        & 3 \\
1      & 1        & 0        & ?   \\
2      & 0        & 1        & ?   \\
3      & ? & ? & 1   \\
},\quad 
\overline{\bB}_2 = 
\kbordermatrix{
       & 1        & 2        & 3 \\
1      & 1        & \rho        & ?   \\
2      & \rho        & 1        & ?   \\
3      & ? & ? & 1   \\
},
\quad \text{and} \quad
\overline{\bB}_3 = 
\kbordermatrix{
       & 1        & 2        & 3 \\
1      & 1        & ?        & ?   \\
2      & ?        & 1        & ?   \\
3      & ? & ? & 1   \\
}.
\end{align*}
\begin{tikzpicture}[remember picture,overlay]
\draw[blue, thick] (2.60,0.85) -- (3.70,0.85) -- (3.70,2.1)--(2.60,2.1)--(2.60,0.85);
\draw[red, ultra thick, dashed] (2.60, 0.2) -- (4.30,0.2) -- (4.30,2.1)--(3.80,2.1)--(3.80,0.75)--(2.60,0.75)--(2.60,0.2);
\node[blue] at (2.5, 2.35) [anchor = center]{$\bB_1$};
\draw[blue, thick] (6.90,0.85) -- (8.00,0.85) -- (8.00,2.1)--(6.90,2.1)--(6.90,0.85);
\draw[red, ultra thick, dashed] (6.90, 0.2) -- (8.55,0.2) -- (8.55,2.1)--(8.10,2.1)--(8.10,0.75)--(6.90,0.75)--(6.90,0.2);
\node[blue] at (6.75, 2.35) [anchor = center]{$\bB_2$};
\draw[red, ultra thick, dashed] (12.35,0.2)--(13.95,0.2)--(13.95,2.1)--(12.35,2.1)--(12.35,0.2);
\end{tikzpicture}

There are two notable features in the examples of $\bBbar_3$ and the other two matrices. First, for any agent $i = 1, 2, 3$ in the network, all entries representing the correlation coefficient between the influence of the new agent $3$ on $i$ and the influence of another existing agent $j = 1, 2$ on $i$ are undetermined; that is, $(\bBbar_i)_{j3}$ is undetermined. In the illustrated matrices, this feature is reflected by all the entries in the third row or column being undetermined and represented by the question mark (i.e., `$?$'). This highlights the first source of uncertainty introduced by the network expansion: the DM lacks information on how the new agent’s influence on a given agent correlates with the influence exerted by other existing agents on that agent.

Second, the correlation coefficients of the existing members' influences on agents are asymmetric. Specifically, the influences of existing members on existing agents $1$ and $2$ are fully known. This is illustrated by the fact that, for agent $i = 1, 2$, the principal submatrix $\bB_i$ of $\bBbar_i$, highlighted by a blue solid-line box in the illustrated matrices, has all its entries determined. In contrast, the influences of the existing agents on the new agent $3$ are unknown. This highlights the second source of uncertainty introduced by the network expansion: the DM lacks information on how the existing agents' influences on the new agent are correlated with one another.

Consequently, all the off-diagonal entries in $\bBbar_3$ are undetermined and represented by the question mark in the illustrated matrices. This implies that, in the worst-case scenario, Nature's choice of $\bBbar_3$ presents the rank-1 structure, as established in \autoref{thm:rank-1}. Therefore, the remaining analysis focuses on characterizing undetermined entries in $\bBbar_i$ for $i = 1, 2$ within an appropriate uncertainty set. 

To model the DM's partial information, we introduce the following notation. Let ${\rm PD}_{k}$ and ${\rm PSD}_{k}$ denote the sets of all $k \times k$ real symmetric positive definite and positive semi-definite matrices, respectively. Without loss of generality, assume that agent $n+1$ is the new agent joining the existing network of $n$ agents. Let $\bB_i \in {\rm PSD}_n$ for $1 \leq i \leq n$ represent the covariance matrix among the existing $n$ agents, and let $\bBbar_i \in {\rm PSD}_{n+1}$ for $1 \leq i \leq n+1$ denote the covariance matrix of all $n+1$ agents. 
Since we are modeling the DM's partial information, we require that for each agent $i = 1, \dots, n+1$, $\bB_i$ is the $n \times n$ principal submatrix of $\bBbar_i$, obtained by selecting the first $n$ rows and columns of $\bBbar_i$. On the other hand, let $b_i = \left( \overline{\bB}_i \right)_{(n+1) \, (n+1)} > 0 $ represent the variance of the link $\bG_{i (n+1)}$ for $i = 1, \dots, n+1$. Given $\bB \in {\rm PSD}_n$ and $b > 0$, we define $\overline{\cB}_{\bB, b}^{\rm PSD}$ as the set of all extended matrices of $\bB$ with its last $(n+1, n+1)$ entry fixed as $b$, forming a new uncertainty set. Specifically, we define $\overline{\cB}_{\bB,b}^{\rm PSD} = \overline{\cB}_{\bB,b} \cap {\rm PSD}_{n+1}$, where
\begin{align}
\overline{\cB}_{\bB,b} &\!=\! \{ \overline{\bB} \!\in\! \R^{(n+1) \times (n+1)} \, | \,\overline{\bB} \text{ is symmetric, }  \overline{\bB}_{ij} \!=\! \bB_{ij} \text{ for all } 1 \!\le\! i,j \!\le\! n, \overline{\bB}_{(n+1)(n+1)} \!=\! b \}. \nn
\end{align}
We denote by $\xbar \in \R^{n+1}$ the intervention of the DM. 

\autoref{thm:partial_info} establishes the uniqueness of Nature's best response with respect to the DM's intervention $\xbar$, that is, the uniqueness of the worst-case scenario.\footnote{For expositional simplicity, we assume here that the DM possesses no knowledge of the new agent $n+1$'s interaction with the existing $n$ agents. As such, \autoref{thm:partial_info} is a special case of \autoref{thm:general_theorem}, in which the DM is allowed to have ``partial'' knowledge of the new agent $n+1$'s interaction with the existing $n$ agents. See \autoref{sec:appendix:proofs} for more details.}

\begin{theorem}[Uniqueness of the Worst-Case Scenario]\label{thm:partial_info}
If $\bB \in {\rm PD}_n$ and $\xbar$ has no zero entry, the worst-case scenario $\bBbar^* \in \overline{\cB}_{\bB,b}^{\rm PSD}$ that maximizes $\langle \xbar, \overline{\bB} \xbar \rangle$ is unique.
\end{theorem}

To provide intuition behind \autoref{thm:partial_info}, let us revisit the example network of two existing members (agent 1 and agent 2) and one new agent (agent 3). By linearity, we have $\langle \xbar, \bBbar \xbar \rangle = \langle \xbar, \bBbar_1 \xbar \rangle + \langle \xbar, \bBbar_2 \xbar \rangle + \langle \xbar, \bBbar_3 \xbar \rangle$. Since none of the entries of $\xbar$ is zero, \autoref{thm:rank-1} shows the unique $\bBbar_3^*$ is given by $\bBbar_3^* = \sigma_3^2 (\bw_3 \otimes \bw_3)$ with $\sigma_3^2 = \sum_{j=1}^3 \bv_{3j}^2 = 3$, $\bw_3 = \frac{\bq_3}{\|\bq_3\|}$, and $\bq_3 = (s(\overline{x}_1), s(\overline{x}_2), s(\overline{x}_3))^\trans$. As a result, $\bBbar_3^*$ is a rank-1 matrix.

On the other hand, the uniqueness of the worst-case scenario for the existing agents $i=1,2$ stems from the strict convexity of the uncertainty set, which contrasts with the reasoning behind the uniqueness of \( \bBbar_3 \). To determine $\bBbar_1^*$, observe that $\langle \xbar, \bBbar_1 \xbar \rangle = 2\overline{x}_3  \left( (\bBbar_1)_{13} \overline{x}_1 + (\bBbar_1)_{23} \overline{x}_2 \right) + {\rm constant}$, where $(\bBbar_1)_{13} = {\rm Corr}(\bG_{11}, \bG_{13}) \in [-1, 1]$ and $(\bBbar_1)_{23} = {\rm Corr}(\bG_{12}, \bG_{13}) \in [-1, 1]$. The dashed square in \autoref{fig:partial_info}-(a) represents these restrictions of the possible values of $(\bBbar_1)_{13}$ and $(\bBbar_1)_{23}$. Additionally, Nature's choice of uncertainty is constrained by the requirement that $\bBbar_1$ must be positive semi-definite, which holds if and only if ${\rm det}(\bBbar_1) \geq 0$ because $\bB_1 \in {\rm PD}_2$ and $(\bBbar_1)_{33} >0$. The Schur Complement theorem \citep{Meyer:Book:2010} states that ${\rm det}(\bBbar_1) \geq 0$ if and only if $(\bBbar_1)_{33} - ((\bBbar_1)_{13}, (\bBbar_1)_{23}) \bB_1^{-1} ((\bBbar_1)_{13}, (\bBbar_1)_{23})^\trans \geq 0$, which is equivalent to $(\bBbar_1)_{13}^2 + (\bBbar_1)_{23}^2 \leq 1$, represented by the gray unit disk in \autoref{fig:partial_info}-(a) contained in the dashed square. Consequently, the uncertainty set is strictly convex and compact.\footnote{Our proof in \autoref{sec:appendix:proofs} addresses a more general condition, and the strict convexity of the uncertainty set is not a direct consequence of the Schur Complement theorem.} The worst-case scenario for agent $1$ must therefore be chosen within this uncertainty set.

\begin{figure}[ht]
\centering
\footnotesize
\subfigure[uncertainty set for agent 1]{
\begin{tikzpicture}[scale=0.54]

\draw[dashed, thick] (2,-2)--(2,2)--(-2,2)--(-2,-2)--(2,-2) {};

\draw (0,2) node[anchor = south east] {$1$};
\draw (0,-2) node[anchor = north east] {$-1$};
\draw (2,0) node[anchor = north west] {$1$};
\draw (-2,0) node[anchor = north east] {$-1$};

\draw[line width = 0.01 in, black!80, dashed, fill = gray!25] (0,0) circle [radius = 2];
\draw[black,thick] (0,0) circle (2);

\draw[thick,->,black] (-2.5,0)--(2.5,0) node[right] {$(\bBbar_1)_{13}$}; 
\draw[thick,->,black] (0,-2.5)--(0,2.5) node[above] {$(\bBbar_1)_{23}$}; 

\draw[ultra thick, blue, ->] (2/1.414, 2/1.414) -- (2/1.414+0.8, 2/1.414+0.8) node[anchor = west] {};
\draw[ultra thick, blue] (2/1.414, 2/1.414) -- (2/1.414+0.2,2/1.414-0.2) -- (2/1.414+0.4,2/1.414) -- (2/1.414+0.2,2/1.414+0.2);
\draw[ultra thick, blue] (2/1.414, 2/1.414) -- (2/1.414-1, 2/1.414+1);
\draw[ultra thick, blue] (2/1.414, 2/1.414) -- (2/1.414+1, 2/1.414-1);
\end{tikzpicture}
}
\hspace{0.02 in}
\subfigure[uncertainty set for agent 2 if $\rho = -0.5$]{
\begin{tikzpicture}[scale=0.54]

\draw[dashed, thick] (2,-2)--(2,2)--(-2,2)--(-2,-2)--(2,-2) {};

\draw (0,2) node[anchor = south east] {$1$};
\draw (0,-2) node[anchor = north east] {$-1$};
\draw (2,0) node[anchor = north west] {$1$};
\draw (-2,0) node[anchor = north east] {$-1$};

\filldraw[rotate around={45:(0,0)}, color=black!80, fill = gray!25, thick] (0,0) ellipse ({2*sqrt(3/2)} and {2*sqrt(1/2)});

\draw[thick,->,black] (-2.5,0)--(2.5,0) node[right] {$(\bBbar_2)_{13}$}; 
\draw[thick,->,black] (0,-2.5)--(0,2.5) node[above] {$(\bBbar_2)_{23}$}; 

\draw[ultra thick, blue, ->] (1.75, 1.75) -- (1.75+1, 1.75+1) node[anchor = west] {};
\draw[ultra thick, blue] (1.75, 1.75) -- (1.75+0.2,1.75-0.2) -- (1.75+0.4,1.75) -- (1.75+0.2,1.75+0.2);
\draw[ultra thick, blue] (1.75, 1.75) -- (1.75-1, 1.75+1);
\draw[ultra thick, blue] (1.75, 1.75) -- (1.75+1, 1.75-1);
\end{tikzpicture}
}
\hspace{0.02 in}
\subfigure[uncertainty set for agent 2 if $\rho = 0.5$]{
\begin{tikzpicture}[scale=0.54]

\draw[dashed, thick] (2,-2)--(2,2)--(-2,2)--(-2,-2)--(2,-2) {};

\draw (0,2) node[anchor = south east] {$1$};
\draw (0,-2) node[anchor = north east] {$-1$};
\draw (2,0) node[anchor = north west] {$1$};
\draw (-2,0) node[anchor = north east] {$-1$};

\filldraw[rotate around={135:(0,0)}, color=black!80, fill = gray!25, thick] (0,0) ellipse ({2*sqrt(3/2)} and {2*sqrt(1/2)});

\draw[thick,->,black] (-2.5,0)--(2.5,0) node[right] {$(\bBbar_2)_{13}$}; 
\draw[thick,->,black] (0,-2.5)--(0,2.5) node[above] {$(\bBbar_2)_{23}$}; 

\draw[ultra thick, blue, ->] (1, 1) -- (1+0.8, 1+0.8) node[anchor = west] {};
\draw[ultra thick, blue] (1, 1) -- (1+0.2,1-0.2) -- (1+0.4,1) -- (1+0.2,1+0.2);
\draw[ultra thick, blue] (1, 1) -- (1-1, 1+1);
\draw[ultra thick, blue] (1, 1) -- (1+1, 1-1);

\end{tikzpicture}
}
\caption{Illustration of the worst-case scenario. In each figure, the gradient of Nature's objective function at the unique worst-case scenario is orthogonal to the uncertainty set denoted by the gray elliptical disk.}\label{fig:partial_info}
\end{figure}
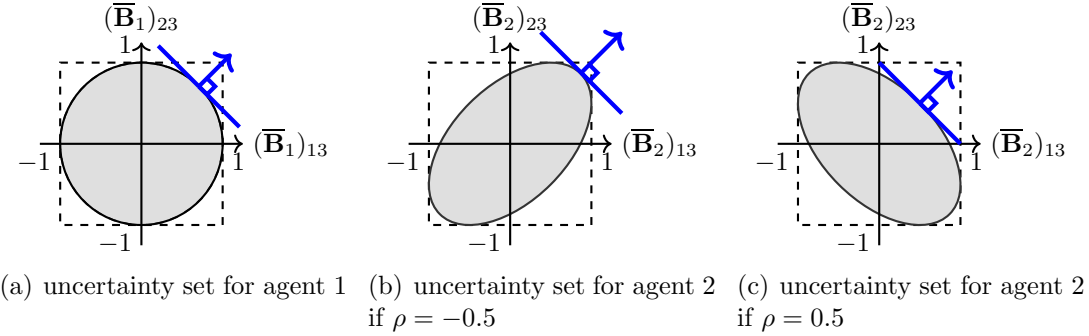

Since the uncertainty set is strictly convex, the gradient of $\langle \xbar, \bBbar_1 \xbar \rangle$ with respect to $\big((\bBbar_1)_{13}, (\bBbar_1)_{23} \big)$ must be orthogonal to the boundary of the uncertainty set. The assumption of non-zero entries in $\xbar$ implies that the gradient (represented by the blue arrow in \autoref{fig:partial_info}) is not vanishing. Consequently, there is a unique optimal choice $\bBbar_1^*$, in which $\nabla \langle \xbar, \overline{\bB}_1 \xbar \rangle \big\vert_{\bBbar_1 = \bBbar_1^*} = 2\overline{x}_3 (\overline{x}_1, \overline{x}_2)$ is orthogonal to the boundary of the uncertainty set as illustrated in \autoref{fig:partial_info}-(a). At the optimal choice, a trade-off arises among the correlation coefficients due to the strict convexity of the uncertainty set, which results from the partial correlation information constraint $\bB_1$.

The worst-case scenario for agent $2$, $\bBbar_2^*$, is also uniquely determined by the strict convexity of the uncertainty set. However, it is worth examining how the shape of uncertainty set changes due to the presence of the correlation among the influences toward agent $2$. The dashed square in \autoref{fig:partial_info}-(b) represents the constraints that $(\bBbar_2)_{13} = {\rm Corr}(\bG_{21}, \bG_{23}) \in [-1, 1]$ and $(\bBbar_2)_{23} = {\rm Corr}(\bG_{22}, \bG_{23}) \in [-1, 1]$. Again, Nature's choice of uncertainty faces another restriction: $\bBbar_2$ must be positive semi-definite, which holds if and only if, by the Schur Complement theorem,
\begin{align*}
\begin{bmatrix}
(\bBbar_2)_{13}  & (\bBbar_2)_{23}
\end{bmatrix}
\begin{bmatrix}
1     & -\rho \\
-\rho & 1 
\end{bmatrix}
\begin{bmatrix}
(\bBbar_2)_{13} \\
(\bBbar_2)_{23}
\end{bmatrix}
\leq 1-\rho^2.
\end{align*}
 The pairs of $\big((\bBbar_2)_{13}, (\bBbar_2)_{23}\big)$ satisfying the above inequality forms an elliptical disk as illustrated in \autoref{fig:partial_info}-(b),(c). For each $\rho \in [-1, 1]$, there are two invariant principal components (i.e., eigenvectors) $(1,1)^\trans$ and $(1, -1)^\trans$, with corresponding non-negative eigenvalues $1 - \rho$ and $1+\rho$, respectively. The uncertainty set for $\rho = -0.5$ is shown as the rotated gray elliptical disk in \autoref{fig:partial_info}-(b), and the uncertainty set for $\rho = 0.5$ is illustrated in \autoref{fig:partial_info}-(c). In both figures, the uncertainty sets are strictly convex.

In summary, the strict convexity of the uncertainty sets yields a unique worst-case scenario $\bBbar_i^*$ for each $i=1,2,3$, and moreover, $\bBbar_2^*$ can vary depending on the correlation coefficient $\rho$. If all the entries of $\overline{\bx}$ are positive, the gradient of Nature's objective function at the worst-case scenario is orthogonal to the uncertainty set and lies in the first quadrant. As a result, the uncertainty set yields lower values of $(\bBbar_2^*)_{13}$ and $(\bBbar_2^*)_{23}$ as the correlation $\rho$ increases, as shown in \autoref{fig:partial_info}-(c). On the other hand, if $\overline{\bx}$ includes a negative entry (e.g., $\overline{\bx}_1 > 0$ and $\overline{\bx}_2 < 0$), the worst-case scenario values for $(\bBbar_2)_{13}$ and $(\bBbar_2)_{23}$ may increase in magnitude as $\rho$ increases, as in \autoref{fig:partial_info}-(b).






\subsection{Higher-Order Interactions}


So far, we have addressed the robust intervention problem in the context of one-shot interactions among agents. However, in many studies in network economics, the focus shifts to settings where agents interact repeatedly or infinitely, leading to long-run equilibrium outcomes. Such scenarios are commonly analyzed using the equilibrium representation $(\bI - \delta \bG)^{-1}$, where $\delta > 0$ captures the strength of the network effect.

If the spectral radius of $\delta \bG$ is less than one---ensured when $\delta$ is sufficiently small—the equilibrium can be expressed as a power series: $(\bI - \delta \bG)^{-1} = \sum_{k=0}^\infty (\delta \bG)^k$. In this representation, $\bG^k_{ij}$ represents the weighted sum of all walks of length $k$ from agent $j$ to agent $i$, capturing the $k$-th order influence of agent $j$'s allocation on agent $i$'s outcome.\footnote{A walk of length $k$ from node $j$ to $i$ is a sequence of nodes $(i_0, i_1, i_2, \dots, i_k)$, such that $i_0 = j$, $i_k = i$, and nodes in the sequence need not be distinct \citep{Jackson:Book:2010}.} This power series representation aggregates the influence dynamics across all possible walks within the network.

The critical question, then, is how the DM can incorporate these higher-order influences and the associated higher-order uncertainties into the design of a robust intervention. By accounting for these dynamics, the DM can address the challenges of designing effective strategies in networks where agents' interactions propagate over time and across multiple pathways. Motivated by this, we now consider an extension of our model to incorporate higher-order uncertainty under additional assumptions. 

First, we assume that with a second-order approximation of $(\bI - \delta \bG)^{-1} \approx (\bI + \delta \bG + \delta^2 \bG^2)$, the DM minimizes the following objective:
\begin{align} 
\E \left[ \| (\bI + \delta \bG + \delta^2 \bG^2) \bx - \bz \|^2 \right]. 
\end{align} 
Second, we assume $\bG_{ii} = 0$ for all $i \in N$. This assumption is often valid in contexts such as network games or supply chain network models, where the zero-order interaction term $\bI$ captures an agent's self-influence after normalization. Third, we assume statistical independence between the rows of \( \bG \), meaning that the influence on an agent is independent of the influence on any other agents. For instance, in a social learning model, agent $i$'s weight on agent $k$ is independent of agent $j$'s weight on $k$. 

Under these assumptions, the adversarial Nature's choice of the worst-case scenario is characterized by a unique rank-1 covariance matrix $\bB_i^*$ for each agent $i$. To simplify exposition, we assume that $\bm_{ij} = \E [ \bG_{ij} ] = 0$ for all $i,j \in N$, allowing us to focus solely on the impact of uncertainty.\footnote{The proof of \autoref{prp:higher_order} does not rely on this zero mean influence of the agents.} Under this assumption, the expected squared distance between agent $i$'s final outcome and the target outcome $z_i$ is 
\begin{align}\label{eqn:higher-order}
\E \left[ | (\bI + \delta \bG + \delta^2 \bG^2)_i \bx - z_i |^2 \right]\! =\! x_i^2 \!+\! \delta \langle \bx, \bB_i \bx \rangle \!+\! \sum_{k=1}^n (\delta \bv_{ik})^2 \langle \bx, \bB_k \bx \rangle \!+\! z_i^2 \!-\! 2 z_i x_i.
\end{align}
The term $\delta \langle \bx, \bB_i \bx \rangle$ in expression \eqref{eqn:higher-order} arises from $(\delta \bG)_i \bx$, capturing the uncertainty generated by the first-order interactions among the agents with respect to agent $i$. 

Another term, $\sum_{k=1}^n (\delta \bv_{ik})^2 \langle \bx, \bB_k \bx \rangle$, accounts for the second-order interactions affecting agent $i$'s final outcome. For these second-order interactions, consider pairs of agents $s$ and $k$ with intermediate agents $s'$ and $k'$ directly influencing agent $i$. The second-order interaction from $s$ to $i$ and from $k$ to $i$ are represented by $\bG_{is'}\bG_{s's}$ and $\bG_{ik'}\bG_{k'k}$, respectively. If $s' \neq k'$, the correlation between these second-order interactions is zero due to the independence in link formation among the agents' interactions, captured by the computation $\E \left[\bG_{is'}\bG_{s's}  \bG_{ik'}\bG_{k'k}\right] = \E \left[\bG_{is'}  \bG_{ik'}\right] \E \left[ \bG_{s's}\right] \E \left[\bG_{k'k} \right] = 0$.

Consequently, the second order effect arises only when there is a common intermediate agent (i.e., $k' = s'$). For each intermediate agent, the correlations among walks of length 2 that share agent $k'$ as the common intermediate agent are amplified by the discounted variance $(\delta \bv_{ik'})^2$. For example, in \autoref{fig:layer}, the solid blue arrows directed toward intermediate agent $2$ influence agent $1$'s outcome and are independent of the other arrows directed toward a distinct agent, such as the dashed green arrows leading to intermediate agent $3$. Then, the impact of the blue arrows on agent $1$ is weighted by the discounted variance term $\delta^2 \bv_{12}^2$, while the impact of the green arrows on agent $1$ is weighted by the discounted variance term $\delta^2 \bv_{13}^2$.    

\begin{figure}[ht]
\centering
\begin{tikzpicture}
[scale=0.45, state/.style={circle, fill = white, draw=black!80, inner sep = 0.02 in, line width = 0.015 in}]

\footnotesize{
\node[rectangle, draw=black, inner sep=0pt, minimum size = 15 pt, fill=blue!10] (a1) at (-10,4)  {$x_1$};
\node[rectangle, draw=black, inner sep=0pt, minimum size = 15 pt, fill=blue!10] (a2) at (-10,2)  {$x_2$};
\node[rectangle, draw=black, inner sep=0pt, minimum size = 15 pt, fill=blue!10] (a3) at (-10,0)  {$x_3$};
\node         (a4) at (-10,-1.5) {$\vdots$};
\node[rectangle, draw=black, inner sep=0pt, minimum size = 15 pt, fill=blue!10] (a5) at (-10,-4)  {$x_n$};

\node[circle, draw=black, inner sep=0pt, minimum size = 15 pt] (b1) at (-5,4)  {1};
\node[circle, draw=black, inner sep=0pt, minimum size = 15 pt] (b2) at (-5,2)  {2};
\node[circle, draw=black, inner sep=0pt, minimum size = 15 pt] (b3) at (-5,0)  {3};
\node                                                          (b4) at (-5,-1.5) {$\vdots$};
\node[circle, draw=black, inner sep=0pt, minimum size = 15 pt] (b5) at (-5,-4)  {$n$};

\node[circle, draw=black, inner sep=0pt, minimum size = 15 pt] (c1) at (0,4)  {1};
\node[circle, draw=black, inner sep=0pt, minimum size = 15 pt] (c2) at (0,2)  {2};
\node[circle, draw=black, inner sep=0pt, minimum size = 15 pt] (c3) at (0,0)  {3};
\node                                                                   (c4) at (5,-1.5) {$\vdots$};
\node[circle, draw=black, inner sep=0pt, minimum size = 15 pt] (c5) at (0,-4)  {$n$};

\node[circle, draw=black, inner sep=0pt, minimum size = 15 pt] (d1) at (5,4)  {$1$};
\node[circle, draw=black, inner sep=0pt, minimum size = 15 pt] (d2) at (5,2)  {$2$};
\node[circle, draw=black, inner sep=0pt, minimum size = 15 pt] (d3) at (5,0)  {$3$};
\node          (d4) at (5,-1.5) {$\vdots$};
\node[circle, draw=black, inner sep=0pt, minimum size = 15 pt] (d5) at (5,-4)  {$n$};

\foreach \from/\to in {a1/b1,a2/b2,a3/b3,a5/b5}
\draw [line width = 0.01 in, ->, thick] (\from) -- (\to);

\foreach \from/\to in {b1/c1,b2/c1,b3/c1,b5/c1}
\draw [line width = 0.01 in, ->, dotted,red] (\from) -- (\to);
\foreach \from/\to in {b1/c2,b2/c2,b3/c2,b5/c2}
\draw [line width = 0.01 in, ->, blue] (\from) -- (\to);
\foreach \from/\to in {b1/c3,b2/c3,b3/c3,b5/c3}
\draw [line width = 0.01 in, ->, dashed,green] (\from) -- (\to);
\foreach \from/\to in {b1/c5,b2/c5,b3/c5,b5/c5}
\draw [line width = 0.01 in, ->, dashdotted,orange] (\from) -- (\to);

\foreach \from/\to in {c1/d1}
\draw [line width = 0.01 in, ->, dotted,red] (\from) -- (\to);
\foreach \from/\to in {c2/d1}
\draw [line width = 0.01 in, ->, blue] (\from) -- (\to);
\foreach \from/\to in {c3/d1}
\draw [line width = 0.01 in, ->, dashed,green] (\from) -- (\to);
\foreach \from/\to in {c5/d1}
\draw [line width = 0.01 in, ->, dashdotted,orange] (\from) -- (\to);

\foreach \from/\to in {c1/d1}
\draw [line width = 0.01 in, ->, dotted,red] (\from) -- (\to);
\foreach \from/\to in {c2/d1}
\draw [line width = 0.01 in, ->, blue] (\from) -- (\to);
\foreach \from/\to in {c3/d1}
\draw [line width = 0.01 in, ->, dashed,green] (\from) -- (\to);
\foreach \from/\to in {c5/d1}
\draw [line width = 0.01 in, ->, dashdotted,orange] (\from) -- (\to);

\node at (7.2,4) [anchor = west] {$x_1^2 + \delta \langle \bx, \bB_1 \bx \rangle + \sum_{k=1}^n (\delta \bv_{1k})^2 \langle \bx, \bB_k \bx \rangle$};

\node at (7.2,2) [anchor = west] {$x_2^2 + \delta \langle \bx, \bB_2 \bx \rangle + \sum_{k=1}^n (\delta \bv_{2k})^2 \langle \bx, \bB_k \bx \rangle$};

\node at (7.2,0) [anchor = west] {$x_3^2 + \delta \langle \bx, \bB_3 \bx \rangle + \sum_{k=1}^n (\delta \bv_{3k})^2 \langle \bx, \bB_k \bx \rangle$};

\node at (7.2,-4) [anchor = west] {$x_n^2 + \delta \langle \bx, \bB_n \bx \rangle + \sum_{k=1}^n (\delta \bv_{nk})^2 \langle \bx, \bB_k \bx \rangle$};

\node at (-10,5.9)  {allocation};
\node at (-10,5.2)  {vector};
\node at (-5,5.9)  {direct};
\node at (-5,5.2)  {influence};
\node at (0,5.9)  {indirect};
\node at (0,5.2)  {influence};
\node at (5,5.9)  {final};
\node at (5,5.2)  {outcome};
\node at (11.5,5.9)  {$\E[$squared distance to};
\node at (12.5,5.2)  {the target outcome$]$};

\draw [dashed, gray, ultra thick] (-2,-5) -- (6.7,-5) -- (6.7,6.5) -- (-2,6.5) -- (-2,-5);
\node at (5,-5) [anchor = north] {\textbf{independence of}};
\node at (5,-5.6) [anchor = north] {\textbf{link formation}};

}

\end{tikzpicture}
\caption{Illustration of second-order uncertainty generation in network. For simplicity, here we assume that $\bz = \mathbf{0}$.}\label{fig:layer}
\end{figure}
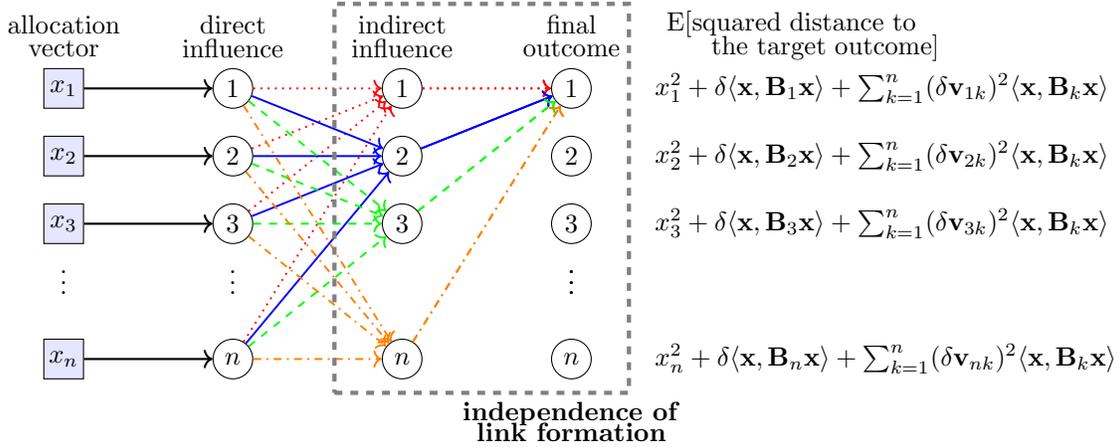

By summing the squared distances for all agents, we obtain 
\begin{align}\label{eqn:second_order}
\E \left[ \| (\bI + \delta \bG + \delta^2 \bG^2) \bx - \bz \|^2 \right] = \langle \bx, \bx \rangle + \sum_{i=1}^n w_i \langle \bx, \bB_i \bx \rangle - 2 \langle \bz , \bx \rangle + \langle \bz, \bz \rangle, 
\end{align} 
where $w_i = \delta + \delta^2 \sum_{k=1}^n \bv_{ki}^2 > 0$. The above expression remains linear in each \( \bB_i \), similar to expression \eqref{eqn:f(x,B)_main} in the main model without higher-order considerations. As detailed in \autoref{thm:rank-1}, this implies that the adversarial Nature’s selection of the worst-case scenario satisfies the rank-1 property for all agents. The following proposition summarizes this result.

\begin{proposition}\label{prp:higher_order}
With the higher-order consideration with the objective function \eqref{eqn:higher-order}, there exists a unique worst-case scenario $\bB_i^*$ for each agent $i \in N$, and it is rank-1.
\end{proposition}
The uniqueness of the worst-case scenario, in turn, allows the DM to solve her robust optimal intervention problem by selecting her intervention based on the first-order condition under the worst-case scenario, in a manner alalogous to \autoref{thm:rank-1}.

We conclude this subsection by noting that for $k$-th order interaction considerations with $k \geq 3$, the linearity in the objective with respet to the uncertainties $(\bB_i)_{i \in N}$ no longer hold, even under the assumption of independent link formation. The reason is that, even with independent link formation, higher-order interactions among influences on an agent can emerge. For example, consider walks of length 3, such as $(3,2,1,2)$ and  $(4,2,1,2)$, which differ only in the initial agent. In this case, the interaction among the influences \( \bG_{23} \), \( \bG_{24} \) and \( \bG_{21} \) in the covariance term may introduce nonlinear effects. As a result, the rank-1 property and the uniqueness of the worst-case scenario may no longer hold. We leave the analysis of higher-order interactions and robust intervention design for future research.





\section{Conclusion}
\label{sec:conclusion}

This paper explores the problem of robust intervention in networks, focusing on a DM’s interaction with agents under uncertainty. We develop a framework that accounts for adversarial uncertainty, modeled as the strategic choices of Nature. By considering the worst-case scenario, the model provides a robust optimization strategy that ensures effective interventions even when the DM lacks precise information about the network's structure.

A central finding of our analysis is the characterization of the DM's unique robust intervention strategy and Nature's optimal adversarial response. We establish that under mild conditions, the worst-case scenario is uniquely determined, enabling the DM to optimize network interventions despite significant uncertainty. Furthermore, we distinguish between the costs of global and local uncertainty, showing that global uncertainty is concentrated in the direction of intervention while local uncertainty propagates through network interactions, amplifying risk. These insights highlight the fundamental trade-off between maximizing influence and mitigating uncertainty-driven disruptions.

We extend the model in two important directions. First, we incorporate higher-order interactions, demonstrating that the rank-1 property of the worst-case scenario persists even when network effects propagate beyond direct interactions. This result ensures that intervention strategies remain predictable in more complex network settings, such as financial contagion or multi-layered supply chains. Second, we examine robust intervention under partial information, focusing on cases where a new agent enters the network. Despite the DM's incomplete knowledge of the new agent’s interactions, we show that a unique worst-case scenario emerges, balancing the DM's existing knowledge with the uncertainty introduced by the new agent.

Our framework contributes to the literature on robust optimization, network games, and robust mechanism design by offering a structured approach for decision-making in uncertain environments. Future work could expand this model to dynamic settings where network structures evolve over time, or explore alternative intervention objectives and the corresponding adversarial responses from Nature.
Additionally, our approach can be extended to the graphon framework introduced by \citet{PariseandOzdaglar:2023:ECMA}, where the adjacency matrix is represented by a graphon (a kernel operator) $\mathcal{W}: [0,1]^2 \to \mathbb{R}$. The uncertainty set for the adversarial Nature can be similarly defined, and the rank-1 property of the worst-case scenario is preserved. This extension allows our model to incorporate continuous network structures, opening new directions for research on robust intervention in large-scale and dynamic networks.



\newpage


\singlespace

\bibliographystyle{apalike}
\bibliography{references} 


\newpage

\appendix

\onehalfspacing

\counterwithin{equation}{section} 
\renewcommand{\theequation}{(\thesection.\arabic{equation})} 


\setcounter{table}{0}
\renewcommand{\thetable}{A\arabic{table}}

\section{Proofs}
\label{sec:appendix:proofs}

 
\subsection*{Proof of \autoref{theorem:duality}}

\begin{proof} 
We first prove the duality \eqref{duality} in \autoref{theorem:duality}. To this end, we introduce the following non-compact minimax theorem \citep{MR638178}:  
\begin{theorem}\label{t2}
Let $X, Y$ be nonempty convex sets, each in a Hausdorff topological vector space, and let $f$ be a real-valued function defined on $X \times Y$ such that (a) for each $y \in Y$, $f(x,y)$ is lower semi-continuous and quasi-convex on $X$, and (b) for each $x \in X$, $f(x,y)$ is upper semi-continuous and quasi-concave on $Y$. Suppose that there exists a non-empty compact set $H \subset X$ and a non-empty compact convex set $K \subset Y$ such that
\begin{align}\label{aa}
\inf_{x\in X}\sup_{y\in Y}f(x,y)\leq\inf_{x\notin H}\max_{y\in K} f(x,y).
\end{align}
Then, it follows that
\begin{align}\label{bb}
\inf_{x\in X}\sup_{y\in Y}f(x,y)=\sup_{y\in Y}\inf_{x\in X}f(x,y).
\end{align}
\end{theorem}

Let $X = \R^n$ and $Y= \cB$ in our model. We can consider $K = Y = \cB$ because $\cB$ is compact and convex. Then, condition \eqref{aa} in the above minimax theorem holds for any compact $H \subset \R^n$. Consequently, property \eqref{bb} follows, which proves the first part of \autoref{theorem:duality}.

We now prove the second part of \autoref{theorem:duality}. We solve Nature's dual problem \eqref{eqn:Nature} by applying the backward induction principle. In the problem, Nature makes a decision, and then the DM decides. Recall that $f(\bx, \bB)$ can be rewritten as 
\begin{align}\label{eqn:plannercost}
f(\bx, \bB) &= \frac{1}{2} \left( \langle \bx, \bM \bx \rangle + \langle \bx, \bB \bx \rangle + \langle \bx, \bC \bx \rangle - 2 \langle \psi^0 + \psi , \bx \rangle + \| \bz \|^2 + \| \bC^{\frac{1}{2}} \bx^0 \|^2  \right) \nn \\
            &= \frac{1}{2} \left( \langle \bx, \bH \bx \rangle - 2  \langle \varphi, \bx \rangle + \| \bz \|^2  + \|\bC^{\frac{1}{2}} \bx^0 \|^2 \right),
\end{align}
where $\bH = \bM + \bB + \bC$ and $\varphi = \psi^0 + \psi$. Note that $f(\bx, \bB)$ is convex in $\bx$. Hence, for a given $\bB$, the DM's optimal choice of $\bx$ must solve the first-order condition, $\frac{\partial f(\bx, \bB)}{\partial \bx} = \bH \bx - \varphi = 0$. Hence, the best response of the DM, $\bx_{BR}(\bB)$, as a function of $\bB$, is $\bx_{BR}(\bB) = \bH^{-1} \varphi$. Given the best-response function, $\bx_{BR}(\bB)$, Nature's objective is to maximize the following with constraint $\bB \in \calB$:
\begin{align}\label{eqn:nature_problem}
f(\bx_{BR}(\bB), \bB)
&= \langle \bx_{BR}(\bB), \bH \bx_{BR}(\bB) \rangle - 2  \langle \varphi , \bx_{BR}(\bB) \rangle + \| \bz \|^2 + \|\bC^{\frac{1}{2}} \bx^0 \|^2  \nn \\
&= - \langle \varphi, \bH^{-1} \varphi \rangle + \| \bz \|^2 + \|\bC^{\frac{1}{2}} \bx^0 \|^2 . 
\end{align}
This shows Nature's objective is to minimize $\langle \varphi, \bH^{-1} \varphi \rangle$ under the constraint $\bB \in \calB$. 

We now state and prove that $\langle \varphi, \bH^{-1} \varphi \rangle$ is convex in $\bB$:
\begin{lemma}\label{lem:convex}
The map $\bB \in \otimes_{i=1}^n \cB_i \mapsto \langle \varphi, (\bM + \bB + \bC)^{-1} \varphi \rangle$ is convex.
\end{lemma}

\begin{proof}[Proof of \autoref{lem:convex}] 
Let $\bB, \bB' \in \otimes_{i=1}^n \cB_i$. For each $t \in [0,1]$, let $\bK_t = \bM + (1-t)\bB + t \bB' + \bC$. Then, it suffices to show that $\langle \varphi, \bK_t^{-1} \varphi \rangle$ is convex in $t$ (i.e., $\langle \varphi, \bK_t^{-1} \varphi \rangle \geq 0$). 

To calculate $\frac{d^2 \bK_t^{-1}}{\partial t^2}$, we first find that by the product rule,
\begin{align*}
0 = \tfrac{d {\bf I}} {d t} = \tfrac{\partial {\bK_t^{-1} \bK_t}} {d t} =  \tfrac{d {\bK_t^{-1}}} {d t} {\bK_t} +  {\bK_t^{-1}}\tfrac{d {\bK_t}} {d t},   
\end{align*}
which implies
\begin{align}\label{eqn:inversediff}
\tfrac{d \bK_t^{-1}}{d t} = - \bK_t^{-1} \tfrac{d \bK_t}{d t} \bK_t^{-1}.
\end{align}
By expression \eqref{eqn:inversediff}, we have
\begin{align*}
\tfrac{d \bK_t^{-1}}{dt}  &= - \bK_t^{-1} ( \bB' - \bB) \bK_t^{-1}, \\
\tfrac{d^2 \bK_t^{-1}}{dt^2} &= - \tfrac{d \bK_t^{-1}}{dt} (\bB' - \bB) \bK_t^{-1} - \bK_t^{-1} (\bB' - \bB)\tfrac{d \bK_t^{-1}}{dt} \\
&= 2 \bK_t^{-1} ( \bB' - \bB) \bK_t^{-1} ( \bB' - \bB) \bK_t^{-1}.
\end{align*}

Let $\rho = ( \bB' - \bB) \bK_t^{-1} \varphi$. Then, we have
\begin{align*}
\tfrac{d^2}{dt^2} \langle \varphi, \bK_t^{-1} \varphi \rangle &= 2 \langle \varphi, \bK_t^{-1} ( \bB' - \bB) \bK_t^{-1} ( \bB' - \bB) \bK_t^{-1} \varphi \rangle \\
&= 2 \langle ( \bB' - \bB) \bK_t^{-1} \varphi, \bK_t^{-1} ( \bB' - \bB) \bK_t^{-1} \varphi \rangle = 2 \langle \rho, \bK_t^{-1} \rho \rangle.
\end{align*}
Since $\bK_t$ is positive definite, $\bK_t^{-1}$ is also positive definite. Thus, $\langle \rho, \bK_t^{-1} \rho \rangle \geq 0$, and so we obtain the desired convexity.
\end{proof}

By the convexity in \autoref{lem:convex}, $\bB^* = \sum_{i=1}^n \bB^*_i$ is optimal for \eqref{eqn:Nature} if and only if 
\begin{align}\label{eqn:KuhnTucker}
 \frac{d}{dt}\bigg|_{t=0} \langle \varphi,( \bM + (1-t)\bB^* + t \bB + \bC )^{-1} \varphi \rangle \geq 0
\end{align}
for any $\bB = \sum_{i=1}^n \bB_i$ for some $\bB_i \in \cB_i$. Using \eqref{eqn:inversediff}, we find
\begin{align*}
&\frac{d}{dt}\bigg|_{t=0} \langle \varphi,(\bM + (1-t)\bB^* + t \bB + \bC )^{-1} \varphi \rangle  \\
&= - \langle \varphi, (\bM+ \bB^* + \bC )^{-1} (\bB - \bB^*) (\bM+ \bB^*+ \bC )^{-1} \varphi \rangle  \\
&= - \langle (\bM+ \bB^*+ \bC )^{-1} \varphi, (\bB - \bB^*) (\bM+ \bB^*+ \bC )^{-1} \varphi \rangle  \!=\! -\langle \bx^* (\bB^*), (\bB - \bB^*) \bx^* (\bB^*) \rangle,
\end{align*}
where $\bx_{BR}(\bB^*) = (\bM+ \bB^* + \bC )^{-1} \varphi$. Note that $\bx_{BR}(\bB^*)$ is the DM's best response to $\bB^*$ in the dual problem \eqref{eqn:Nature}. Thus, showing expression \eqref{eqn:KuhnTucker} is equivalent to prove 
\begin{align}\label{eqn:KuhnTucker2}
\langle \bx_{BR}(\bB^*), \bB \bx_{BR}(\bB^*) \rangle \le \langle \bx_{BR}(\bB^*), \bB^* \bx_{BR}(\bB^*) \rangle \ \text{ for all } \, \bB \in \calB.
\end{align}

Recall that Nature's objective is to maximize $f(\bx, \bB)$ in expression \eqref{eqn:plannercost}:
\begin{align}
f(\bx, \bB) = \frac{1}{2} \left( \langle \bx, \bM \bx \rangle + \langle \bx, \bB \bx \rangle + \langle \bx, \bC \bx \rangle - 2 \langle \varphi , \bx \rangle + \| \bz \|^2 + \| \bC^{\frac{1}{2}} \bx^0 \|^2  \right).
\end{align}
Note that the only term that depends on Nature's choice $\bB$ is $\langle \bx, \bB \bx \rangle$. Hence, the optimality condition \eqref{eqn:KuhnTucker2} means that $\bB^*$ is Nature's best response to $\bx^*(\bB^*)$. Since $\bx_{BR}(\bB^*)$ is the best response to $\bB^*$ by construction, we have established the following:
\begin{claim}\label{claim:equivalence1}
$\bB^*$ is optimal solution to Nature's dual problem \eqref{eqn:Nature} if and only if $(\bx_{BR}(\bB^*), \bB^*)$ is a Nash equilibrium.
\end{claim}

To complete the proof of the theorem, it remains to prove the following:
\begin{claim}\label{claim:DMoptimal}
$\bx^* = \bx_{BR} (\bB^*)$ for any optimal $\bB^*$ for the dual problem \eqref{eqn:Nature}.
\end{claim}
To see why \autoref{claim:DMoptimal} yields the second part of \autoref{theorem:duality}, suppose that $(\bx^*, \bB^*)$ is a Nash equilibrium. $\bB^*$ is a solution to the dual problem \eqref{eqn:Nature} by \autoref{claim:equivalence1}. Since $\bx^*$ is a best response to $\bB^*$, $\bx^* = \bx^*$ by \autoref{claim:DMoptimal}. Conversely, suppose that $\bB^*$ is an optimal solution to the dual problem \eqref{eqn:Nature} and let $\bx^* = \bx^*$. \autoref{claim:DMoptimal} states that $\bx^* = \bx_{BR} (\bB^*)$ is a best response to $\bB^*$. Also, the optimality of $\bB^*$ to the dual problem \eqref{eqn:Nature} implies that $\bB^*$ is a best response to $\bx^*$. Thus, $(\bx^*, \bB^*)$ is a Nash equilibrium by \autoref{claim:equivalence1}. 

Now to show \autoref{claim:DMoptimal}, it suffices to show the following claim:
\begin{claim}\label{claim:claim3}
Let $\xi: \calB \to \R$ be a function defined as $\xi(\bB) = f(\bx_{BR}^*(\bB^*), \bB)$, where $\bB^*$ is an optimal solution to problem \eqref{eqn:Nature}. Then, $\bB^*$ returns the maximum value of function $\xi$.
\end{claim}
Note that \autoref{claim:claim3} is equivalent to show that $(\bx_{BR}^*(\bB^*), \bB^*)$ is a Nash equilibrium, which is already established in \autoref{claim:equivalence1}. 

Finally, observe that 
\begin{align*}
\max_{\bB \in \cB} f(\bx^*, \bB) &= \min_{\bx \in \R^n} \, \max_{\bB \in \cB} f(\bx, \bB) \quad \text{by definition of $\bx^*$} \nn\\
    &= \max_{\bB \in \cB} \, \min_{\bx \in \R^n} f(\bx, \bB) \quad \text{by the first part of \autoref{theorem:duality}} \nn \\
    &= f(\bx_{BR}^*(\bB^*), \bB^*) \quad \text{by definition of $\bB^*$} \nn  \\
    &= \max_{\bB \in \cB} f(\bx_{BR}^* (\bB^*), \bB) \quad \text{by \autoref{claim:claim3}.}
\end{align*}
Consequently, the uniqueness of $\bx^*$ implies that $\bx^* = \bx_{BR}^*(\bB^*)$. 
\end{proof}


\subsection*{Proof of \autoref{thm:rank-1}}

\begin{proof}
We first show all entries of $\bB_i^*$ are extreme values. Since $|\bB_i| \leq \bv_{ij} \bv_{ik}$, 
\begin{align*}
\langle \bx^*, \bB_i \bx^* \rangle = \underbrace{\sum_{l=1}^n (\bB_i)_{ll} (x^*_l)^2}_{\text{constant}} + 2 \sum_{j<k}^n (\bB_i)_{jk} x^*_j x^*_k \leq  \sum_{l=1}^n (\bB_i)_{ll} (x^*_l)^2  + 2 \sum_{j<k}^n ( \bv_{ij} \bv_{ik} ) \, | x^*_j | \, | x^*_k |. 
\end{align*}
Hence, it suffices to find that $\bB_i^*$ achieves the upper bound of $\langle \bx^*, \bB_i \bx^* \rangle$. Define $\bB_i^*$:
\begin{align*}
(\bB_i^*)_{jk} = 
\begin{cases}
\bv_{ij}^2                 &\quad \text{if $j = k$,} \\
s(x^*_j) \, s(x^*_k) \, \bv_{ij} \bv_{ik} &\quad \text{if $j \neq i$,} 
\end{cases}
\end{align*}
Then, $\bB_i^*$ clearly achieves the upper bound value. 

We now show that $\bB_i^*$ is a rank-1 matrix. To show that $\bB^*$ is a rank-1 matrix, it suffices to show that any $k$th row of $\bB_i^*$ with $k \geq 2$ is a multiple of the first row of $\bB_i^*$. Without loss of generality, let us assume that $i = 1$. Then, the $k$th row $(\bB_1^*)_k$ is
\begin{align*} 
&(\bB_1^*)_k = (s(\bx^*_k) s(\bx^*_1) \bv_{1k} \bv_{11}, s(\bx^*_k) s(\bx^*_2) \bv_{1k} \bv_{12}, \dots, s(\bx^*_k) s(\bx^*_n) \bv_{1k} \bv_{1n}) \\
        &= \frac{s(\bx^*_1)}{s(\bx^*_1)} \frac{s(\bx^*_k)}{s(\bx^*_k)} \frac{\bv_{1k}}{\bv_{11}}  \frac{\bv_{11}}{\bv_{1k}} (s(\bx^*_k) s(\bx^*_1) \bv_{1k} \bv_{11}, s(\bx^*_k) s(\bx^*_2) \bv_{1k} \bv_{12}, \dots, s(\bx^*_k) s(\bx^*_n) \bv_{1k} \bv_{1n})    \\
        &= \underbrace{\left( \frac{s(\bx^*_k)}{s(\bx^*_1)} \frac{\bv_{1k}}{\bv_{11}} \right)}_{= \text{$\kappa$ with $\kappa \neq 0$}} (s(\bx^*_1) s(\bx^*_1) \bv_{11} \bv_{11}, s(\bx^*_1) s(\bx^*_2) \bv_{11} \bv_{12}, \dots, s(\bx^*_1) s(\bx^*_n) \bv_{11} \bv_{1n})= \kappa (\bB_1^*)_1.
\end{align*}
Therefore, $\bB_1$ is a rank-1 matrix. 
Furthermore, since the trace of $\bB_i^*$ is the sum of eigenvalues, its unique non-zero eigenvalue must be equal to the sum of variances, which is strictly positive. This also implies that $\bB_i^* \in \calB_i$ as $\bB_i^*$ is positive semi-definite.

Note that the construction of $\bB_i^*$ is independent of specific values of $\bx^*$. In addition, by construction for the off-diagonal entries of $\bB_i^*$, the eigenvector corresponding to the unique non-zero eigenvalue is $\bq = (q_1, \dots, q_n)^\trans$, where $q_k = s(y_k) \bv_{ik}$ for all $k$. Hence, $\bq$ is in the same orthant with $\bx^*$. 

We now show that Nature's optimal choices of rank-1 covariance matrix $\bB_i$ for agent $i \in N$ can be realized under the symmetry assumption on the network $\bG$. Note that for given $i,j \in N$, the symmetry assumption, $\bG_{ij} = \bG_{ji}$, implies that $\bm_{ij} = \bm_{ji}$, $\bv_{ij} = \bv_{ji}$, and $\bU_{ij} = \bU_{ji}$. Due to the uniqueness of $\bB_i^*$ for $i$, we first construct $\bB_i^*$ for $i \in N$ that achieves the upper bound of Nature's objective function. Let $\bX$ be a random variable of mean zero and variance one (e.g., the standard normal random variable). Define $\bU_{ij} = s(x_i^*) s(x_j^*) \bv_{ij} \bX$ for $i,j \in N$. With this definition, we find that $\Var(\bU_{ij}) = \bv_{ij}^2$ for all $i,j \in N$. Let $\bB_i^*$ be the matrix defined as $(\bB_i^*)_{jk} = \E[\bU_{ij} \bU_{ik}]$. Then, as in the proof without symmetry assumption, we see that $(\bB_i^*)$ is Nature's optimal choice agent $i$ under the variance constraint. 

Furthermore, we find that for each $i,j \in N$, $\bU_{ij} = s(x_i^*) s(x_j^*) \bv_{ij} \bX = s(x_j^*) s(x_i^*) \bv_{ji} \bX = \bU_{ji}$, which implies that $\bG_{ij} = \bm_{ij} + \bU_{ij} = \bm_{ji} + \bU_{ji} = \bG_{ji}$. Therefore, the theorem holds under the symmetric assumption about $\bG$.
\end{proof}


\subsection*{Proof of \autoref{prp:vofi}}

\begin{proof}
The proof follows from a direct calculation:
\begin{align*}
\tfrac{\partial f(\bx^*(\bB), \bB)}{\partial \bB}\Big\vert_{\bB = \bB^*} &= \tfrac{\partial f(\bx^*(\bB), \bB)}{\partial \bx} \tfrac{\partial \bx^*}{\partial \bB} \Big\vert_{\bB = \bB^*} + \underbrace{ \tfrac{\partial f(\bx^*(\bB), \bB)}{\partial \bB}\Big\vert_{\bB = \bB^*} }_{\text{$=0$ by the envelope theorem}}
    = \bx^*(\bB^*) \otimes \bx^*(\bB^*).
\end{align*}
Therefore, the proposition is proven.
\end{proof}


\subsection*{Proof of \autoref{prp:vofli}}

\begin{proof}
Since $f(\bx^*, \bB^*) = \langle \bx^*, \bB^* \bx^* \rangle$, where $\bB^* = \sum_{i=1}^n \bB_i^*$ and $(\bB_i^*)_{jk} = \bv_{ij} \bv_{ik} s(x_j) s(x_k)$, we have 
\begin{align*}
f(\bx^*, \bB^*) &= \langle \bx^*, \bB^* \bx^* \rangle = \sum_{i=1}^n \langle \bx^*, \bB_i^* \bx^* \rangle \\
&= \sum_{i=1}^n x_j \left( \sum_{1 \leq j,k \leq n} \bv_{ij} \bv_{ik} \, s(x_j) s(x_k) \right) x_k = \sum_{i=1}^n \left( \sum_{1 \leq j,k \leq n} \bv_{ij} \bv_{ik} \, |x_j| |x_k| \right).
\end{align*}
Consequently, the partial derivative of $f(\bx^*, \bB^*)$ with respect to $\bv_{ij}$ is
\begin{align*}
\frac{\partial f(\bx^*, \bB^*)}{\partial \bv_{ij}} = 2 \left( \sum_{k=1}^n \bv_{ik} \, |x_j| |x_k| \right).
\end{align*}
Therefore, the proposition is proven.
\end{proof}


\subsection*{Proof of \autoref{prp:threshold_v}}

\begin{proof}
We first find a closed-form solution of $\bx^*$ with the parametric assumptions. Then, we calculate the threshold $\overline{v}$. Under the assumption of the parameters, we have
\begin{align*}
\bM =  2
\begin{bmatrix}
m^2 & m \\
m & 1 
\end{bmatrix}
\quad \text{and} \quad \bB 
= 2
\begin{bmatrix}
v^2              & \frac{\rho_1 + \rho_2}{2} \\
\frac{\rho_1 + \rho_2}{2}  &      1
\end{bmatrix}.
\end{align*}
We first guess that $\bx^* \in \R_{++}^2$. Then, $\rho_1^* = \rho_2^* = v$, which implies that
\begin{align*}
\bx^* &= \left( \bM + \bB^* + c \bI \right)^{-1} \varphi 
=
\begin{bmatrix}
m^2 + v^2 + \frac{c}{2} & m + v \\
m+v     & 2 + \frac{c}{2} 
\end{bmatrix}^{-1}
\begin{bmatrix}
m \\
1
\end{bmatrix}
=
\frac{1}{\triangle}
\begin{bmatrix}
(1 + \frac{c}{2})m - v \\
v^2 + \frac{c}{2} - mv
\end{bmatrix},
\end{align*}
where $\triangle$ is the determinant of $(\bM + \bB^* + c \bI)$ that is strictly positive. $\bx_1^* \geq \bx_2^*$ if and only if
\begin{align*}
&\!-\!v^2 \!-\! (1\!-\!m) v + \left( m \!+\! \tfrac{c}{2} m \!-\! \tfrac{c}{2}  \right) \!\geq\! 0
\!\iff\! \!-\! \left( v \!-\! \left( \tfrac{(m-1)}{2} \right)  \right)^2 \!+\! \underbrace{\left( \tfrac{m-1}{2} \right)^2 \!+\! \left( m \!+\! \tfrac{c}{2} (m\!-\!1) \right)}_{> 0 \text{ if $m \geq 1$}} \!\geq\! 0.
\end{align*}
Note that the maximum value of the last expression is strictly greater than 0 if $m \geq 1$. Hence, it suffices to show that the last expression at $v = 0$ is strictly greater than 0 as the expression is strictly concave and quadratic in $v$. When $v = 0$, the expression has the value $\left( m + \frac{c}{2} (m-1) \right) > 0$ if $m \geq 1$. Consequently, there exists $\overline{v} > 0$ such that $x_1^* \geq x_2^*$ if and only if $v \leq \overline{v}(m)$ by the intermediate value theorem.

The closed-form expression of $\overline{v}(m)$ is
\begin{align*}
\overline{v}(m) = \tfrac{1}{2} \left( (m-1) + \sqrt{2 c (m - 1) + (m + 1)^2 } \right),
\end{align*}
and its first and second derivatives are
\begin{align*}
\tfrac{\partial \overline{v}(m)}{\partial m} = \tfrac{1}{2} + \tfrac{(m+c+1)}{2 \sqrt{2 c (m-1) + (m + 1)^2}} > 0,\;\; 
\tfrac{\partial^2 \overline{v}(m)}{\partial m^2} = 
- \tfrac{c (c + 4)}{2 \left( 2 c (m -1) + (m+1)^2 \right)^{3/2}} < 0,
\end{align*}
which prove that $\overline{v}(m)$ is strictly increasing and concave. 
\end{proof}


\subsection*{Proof of \autoref{thm:partial_info}}

We state and prove a generalized version. Recall that ${\rm PD}_{n}$ and ${\rm PSD}_{n}$ denote the sets of all $n \times n$ real symmetric positive definite and positive semi-definite matrices.

\begin{theorem}\label{thm:general_theorem}
Let $\bB \in \rm{PD}_n$, $I$ be a proper (possibly empty) subset of $N = \{1,...,n\}$, $b_i \in \R$ for all $i \in I$, and $b_{n+1} > 0$. Then, for any intervention $\xbar \in \R^{n+1}$ with $\overline{x}_{j} \overline{x}_{n+1} \neq 0$  for some $j \in \{1,...,n\} \setminus I$, the problem
\begin{align}\label{partiallyknownproblem}
\text{maximize} \, \langle \xbar, \overline{\bB} \xbar \rangle \ \text{ over } \  \{  \overline{\bB} \in \overline{\cB}_{\bB, b}^{\rm PSD} \, | \, \overline{\bB}_{i, n+1} = b_i \text{ for all } i \in I\}\end{align}
has a unique solution (unless the latter feasible set is empty). 
\end{theorem}
\autoref{thm:partial_info} is clearly a special case of \autoref{thm:general_theorem} with $I = \emptyset$. We shall prove \autoref{thm:general_theorem}.

\begin{proof}
For a given matrix $\bB \in \R^{n \times n}$, let $\overline{\bB} \in \R^{(n+1) \times (n+1)}$ denote a symmetric and extended matrix of $\bB$ having $\bB$ as its principal submatrix. For each $\bv = (\bv_1,...,\bv_n)^\trans \in \R^n$, let $\overline{\bv} = (\overline{\bv}_1, \dots, \overline{\bv}_n, \overline{\bv}_{n+1})^\trans \in \R^{n+1}$ denote the extended vector of $\bv$ having $\bv$ as its subvector; that is, $\overline{\bv}_k = \bv_k$ for all $1 \leq k \leq n$. By abusing notation, we denote $\overline{\bv}$ by $\overline{\bv} = (\bv^\trans, \bv_{n+1})^\trans$. For a given matrix $\bA$, we let $\text{Null}(\bA)$ denote the \text{Null} space of $\bA$; that is, $\text{Null}(\bA) = \{ \bx \, | \, \bA \bx = \mathbf{0} \}$.

Let $V$ be an affine subspace of $\R^d$ that inherits the usual topology of $\R^d$; that is, $A \subset V$ is open in $V$ if and only if there exists an open set $\overline{A}$ in $\R^d$ such that $A = \overline{A} \cap V$. We say that a compact convex set $A \subset V$ is strictly convex in $V$ if $A$ has nonempty interior in $V$, and for any distinct $\bx$ and $ \bx'$ in $ A$, $\frac{\bx+\bx'}{2}$ is also contained in the interior of $A$. 

Given $\bB \in {\rm PSD}_n$ and $b > 0$, we define $\overline{\cB}_{\bB,b}$ as the set of all extended matrices of $\bB$ having the $(n+1)$-th row and $(n+1)$-th column entry by $b$:
\begin{align*}
\overline{\cB}_{\bB,b} \!&=\! \{ \overline{\bB} \in \R^{(n+1) \times (n+1)} \, | \,\overline{\bB} \text{ is symmetric, }  \overline{\bB}_{ij} = \bB_{ij} \text{ for all } 1 \le i,j \le n, \overline{\bB}_{(n+1)(n+1)} = b \}.
\end{align*}
We define $\overline{\cB}_{\bB,b}^{\rm PSD} = \overline{\cB}_{\bB,b} \cap {\rm PSD}_{n+1}$ and $\overline{\cB}_{\bB,b}^{\rm PD} = \overline{\cB}_{\bB,b} \cap {\rm PD}_{n+1}$.

\autoref{thm:general_theorem} will be shown as a consequence of the following lemma.
\begin{lemma}\label{strictconvexity}
For any $\bB \in {\rm PD}_n$ and $b > 0$, the set 
\begin{align}
{\cal X}_{\bB,b}^{\rm PSD} = \{ \bb \in \R^n \, | \, \bb_i = \overline{\bB}_{i, n+1} \text{ for all } i=1,...,n \text{ for some } \overline{\bB} \in \overline{\cB}_{\bB,b}^{\rm PSD}  \}
\end{align}
is a strictly convex and compact subset of $\R^n$.
\end{lemma}
\begin{proof}
It is clear that ${\cal X}_{\bB,b}^{\rm PSD}$ is compact and convex. For $\bB \in {\rm PD}_n$ and $b > 0$, by Sylvester's criterion \citep{Meyer:Book:2010}, $\overline{\bB} \in \overline{\cB}_{\bB,b}$ is positive semi-definite if and only if ${\rm det}(\overline{\bB}) \geq 0$. Moreover, $\overline{\bB} \in \overline{\cB}_{\bB,b}$ is positive definite if and only if ${\rm det}(\overline{\bB}) > 0$. Thus,
\begin{align*}
{\cal X}_{\bB,b}^{\rm PSD} = \{ \bb \in \R^n \, | \, \bb_i = \overline{\bB}_{i, n+1} \text{ for all } i \in N \text{ for some } \overline{\bB} \in \overline{\cB}_{\bB,b} \text{ with } {\rm det}(\overline{\bB}) \ge 0  \},
\end{align*}
and the interior of the set ${\cal X}_{\bB,b}^{\rm PSD}$ is given by
\begin{align*}
{\cal X}_{\bB,b}^{\rm PD} &= \{ \bb \in \R^n \, | \, \bb_i = \overline{\bB}_{i, n+1} \text{ for all } i \in N \text{ for some } \overline{\bB} \in \overline{\cB}_{\bB,b} \text{ with } {\rm det}(\overline{\bB}) > 0  \} \\
&= \{ \bb \in \R^n \, | \, \bb_i = \overline{\bB}_{i, n+1} \text{ for all } i=1,...,n \text{ for some } \overline{\bB} \in \overline{\cB}_{\bB,b}^{\rm PD} \},
\end{align*}
which is nonempty because $\mathbf{0} \in {\cal X}_{\bB,b}^{\rm PD}$. Let $\bb^1, \bb^2 \in {\cal X}_{\bB,b}^{\rm PSD} $ with $\bb^1 \ne \bb^2 $. We need to show $\frac{\bb^1 + \bb^2}{2} \in {\cal X}_{\bB,b}^{\rm PD}$. For $j=1,2$, let $\overline{\bB}^j \in  \overline{\cB}_{\bB,b}^{\rm PSD}$ such that $\bb^j_i = \overline{\bB}^j_{i, n+1}$ for all $ i=1,...,n$. Then, the claim $\frac{\bb^1 + \bb^2}{2} \in {\cal X}_{\bB,b}^{\rm PD}$ is equivalent to  ${\rm det}( \frac{\overline{\bB}^1 + \overline{\bB}^2}{2}) > 0$. By a way of contradiction, suppose that this is false. Then, ${\rm det}( \frac{\overline{\bB}^1 + \overline{\bB}^2}{2}) = 0$, i.e., $ \frac{\overline{\bB}^1 + \overline{\bB}^2}{2}$ is not positive definite, which is the case if and only if $\text{Null}(\bBbar^1) \cap \text{Null}(\bBbar^2) \neq \{ \mathbf{0} \}$. For any $\overline{\bz} \in \text{Null}(\bBbar^1) \cap \text{Null}(\bBbar^2)$ with $\overline{\bz} \neq \mathbf{0}$, we claim that $\overline{\bz}_{n+1} = 0$. To prove this, suppose $\overline{\bz}_{n+1} \neq 0$ and $\bBbar^j \overline{\bz} = \mathbf{0}$ for $j=1,2$. Since $\bBbar^j \overline{\bz} = \sum_{k=1}^{n+1} \overline{\bz}_k \bBbar^j_k$, where $\bBbar^j_k$ denotes the $k$th column of $\bBbar^j$, the equation $ \sum_{k=1}^{n+1} \overline{\bz}_k \bBbar^j_k= \mathbf{0} \in \R^{n+1}$ in particular yields, by ignoring the last $(n+1)$th row, $\sum_{k=1}^n \overline{\bz}_k  \bB_k + \overline{\bz}_{n+1}\bb^j = \mathbf{0} \in \R^n$, $j=1,2$. However, this yields $\bb^1 = -\sum_{k=1}^{n} \tfrac{\overline{\bz}_k}{\overline{\bz}_{n+1}} \bB_k = \bb^2$, a contradiction. Hence $\overline{\bz}_{n+1} = \mathbf{0}$, but then $\sum_{k=1}^n \overline{\bz}_k  \bB_k= \mathbf{0}$ implies $\overline{\bz} = \mathbf{0}$ due to the assumption $\bB \in {\rm PD}_n$. We conclude $\text{Null}(\bBbar^1) \cap \text{Null}(\bBbar^2) = \{ \mathbf{0} \}$, which yields ${\rm det}( \frac{\overline{\bB}^1 + \overline{\bB}^2}{2}) > 0$, as desired.
\end{proof}

We now prove the theorem. As the only unknown entries in $\overline{\bB}$ are $(\overline{\bB}_{i, n+1})_{i \in \{1,...,n\} \setminus I }$, the problem \eqref{partiallyknownproblem} is equivalent to
$\displaystyle \max_{\overline{\bB} \in \overline{\cB}_{\bB, b}^{\rm PSD}}  \sum_{i \in \{1,...,n\} \setminus I}\xbar_i \xbar_{n+1}  \overline{\bB}_{i, n+1}$. The set of feasible variables $(\overline{\bB}_{i, n+1})\in \R^{n - |I|}$ is equal to the projection of the slice set ${\cal X}_{\bB,b}^{\rm PSD} \cap \{\bx \in \R^n \, | \, \bx_i = b_i \text{ for all } i \in I\}$ onto $\{(\bx_i)_{i \in \{1,...,n\} \setminus I } \} \cong \R^{n - |I|}$. Since the objective function is nonzero and linear in the variable $(\overline{\bB}_{i, n+1})$, and any slice of a strictly convex set is also strictly convex, the theorem follows from \autoref{strictconvexity}.
\end{proof}


\subsection*{Proof of \autoref{prp:higher_order}}

\begin{proof}
For simplicity, we prove the proposition by assuming that the mean influence of the link $\bG_{ij}$ is zero for all $i,j \in N$. Then, we provide a general proof without the assumption.

For each $i \in N$, we find that
\begin{align}\label{eqn:highorder}
&| (\bI + \delta \bG + \delta^2 \bG^2)_i \bx - z_i |^2  = \Big| x_i + \delta \sum_{j=1}^n \bG_{ij} x_j + \delta^2 \sum_{1 \leq k,l \leq n} \bG_{ik} \bG_{kl} x_l - z_i \Big|^2 \nn \\
&= x_i^2 + \left( \delta \sum_{j=1}^n \bG_{ij} x_j \right)^2 + \left( \delta^2 \sum_{1 \leq k,l \leq n} \bG_{ik} \bG_{kl} x_l \right)^2 + z_i^2 + 2 x_i \left( \delta \sum_{j=1}^n \bG_{ij} x_j \right)\nn \\
&\quad  + 2 x_i \left( \delta^2 \sum_{1 \leq k,l \leq n} \bG_{ik} \bG_{kl} x_l \right) - 2 x_i z_i  + 2 \left( \delta \sum_{j=1}^n \bG_{ij} x_j \right) \left( \delta^2 \sum_{1 \leq k,l \leq n} \bG_{ik} \bG_{kl} x_l \right)\nn \\ 
&\quad - 2 z_i \left( \delta \sum_{j=1}^n \bG_{ij} x_j \right) - 2 z_i \left( \delta^2 \sum_{1 \leq k,l \leq n} \bG_{ik} \bG_{kl} x_l \right).
\end{align}
We investigate the expectation of each term in expression \eqref{eqn:highorder}. First, we find that 
\begin{align*}
\E \left[ \left( \sum_{1 \leq k,l \leq n} \bG_{ik} \bG_{kl} x_l \right) \right] &= 0, \quad 
\E \left[ \left( \sum_{j=1}^n \bG_{ij} x_j \right) \left( \sum_{1 \leq k,l \leq n} \bG_{ik} \bG_{kl} x_l \right) \right] = 0, \\
\E \left[ \left( \sum_{j=1}^n \bG_{ij} x_j \right) \right] &= 0, \quad
\E \left[ \left( \sum_{1 \leq k,l \leq n} \bG_{ik} \bG_{kl} x_l \right) \right] = 0.
\end{align*}
Second, $\E \left[ \left( \sum_{j=1}^n \bG_{ij} x_j \right)^2 \right] = \langle \bx, \bB_i \bx \rangle$ as in the main model. Third, we observe that
\begin{align*}
&\E \left[ \left( \sum_{1 \leq k,l \leq n} \bG_{ik} \bG_{kl} x_l \right)^2 \right] = \E \left[(\sum_{1 \leq k,l,s,t \leq n} \bG_{ik} \bG_{kl} \bG_{is} \bG_{st} x_l x_t) \right] \\
&\!=\! \sum_{k \neq i, s \neq i,l,t} \E \left[ \bG_{ik} \bG_{is} \right] \E \left[ \bG_{kl} \bG_{st} \right] x_l x_t 
\!=\! \sum_{k \neq i,l,t} \E \left[ \bG_{ik}^2 \right] \E \left[ \bG_{kl} \bG_{kt} \right] x_l x_t \!=\! \sum_{k=1}^n \bv_{ik}^2 \langle \bx, \bB_k \bx \rangle.
\end{align*}
Thus, we obtain expression \eqref{eqn:higher-order} in the main text:
\begin{align}\label{eqn:second_order_a}
\E \left[ | (\bI + \delta \bG + \delta^2 \bG^2)_i \bx - z_i |^2 \right] &= x_i^2 + \E \left[  \left( \delta \sum_{j=1}^n \bG_{ij} x_j \right)^2 + \left( \delta^2 \sum_{1 \leq k,l \leq n} \bG_{ik} \bG_{kl} x_l \right)^2 \right] + z_i^2 \nn \\
    &= x_i^2 + \delta \langle \bx, \bB_i \bx \rangle + \delta^2 \sum_{k=1}^n \bv_{ik}^2 \langle \bx, \bB_k \bx \rangle + z_i^2.
\end{align}

Finally, by summing the squared distances for all agents, we obtain 
\begin{align}\label{eqn:second_oder_b}
\E \left[ \| (\bI + \bG + \bG^2) \bx - \bz \|^2 \right] &= \sum_{i=1}^n \left( x_i^2 + \delta \langle \bx, \bB_i \bx \rangle + \delta^2 \sum_{k=1}^n \bv_{ik}^2 \langle \bx, \bB_k \bx \rangle + z_i^2 \right) \nn \\
    &= \langle \bx, \bx \rangle + \sum_{i=1}^n w_i \langle \bx, \bB_i \bx \rangle - 2 \langle \bz , \bx \rangle + \langle \bz, \bz \rangle, 
\end{align} 
where $w_i = \delta + \delta^2 \sum_{k=1}^n \bv_{ki}^2 > 0$ as in expression \eqref{eqn:second_order} in the main text. 

Since the DM's objective function \eqref{eqn:second_oder_b} is linear in each $\bB_i$, the rank-1 property and the uniqueness of $\bB_i^*$ are held by \autoref{theorem:duality} and \autoref{thm:rank-1}. Therefore, the worst-case scenario $\bB^* = \sum_{i=1}^n \bB_i^*$ is unique. 

We now prove the proposition without the assumption. To economize notation, without loss of generality, we let $\delta = 1$ because it does not affect the linearity of the DM's objective function. We let $\overline{\bG} = \E[\bG]$ and write $\bG = \overline{\bG} + \bU$. We denote by $\bm_i$ and $ \bm^i$ the $i$th row vector and column vector of $\overline{\bG}$, respectively. Recall that $\bM_i = \bm_i \otimes \bm_i$, and $\E[\bU_i \bU_i^T] = \bB_i$. Since we assume that $\bG_{ii} = 0$, we have $\bU_{ii} = \bm_{ii} = 0$. We calculate that for each $i \in N$,
\begin{align*}
\E \left[ \left( \sum_{j=1}^n \bG_{ij} x_j \right)^2 \right] &= \E \left[ \left((\bm_i + \bU_i)^\trans \bx\right)^2\right] = \langle \bx, \bB_i \bx \rangle + \langle \bx, \bM_i \bx \rangle.
\end{align*}
In addition, it follows that for each $i \in N$,
\begin{align*}
\E \left[ \left( \sum_{1 \leq k,l \leq n} \bG_{ik} \bG_{kl} x_l \right)^2 \right] &= \E \left[\sum_{1 \leq k,l,s,t \leq n} \bG_{ik} \bG_{kl} \bG_{is} \bG_{st} x_l x_t \right] \nn  \\
&= \sum_{1 \leq l,t \leq n} \sum_{k \neq i, s \neq i} \E \left[ \bG_{ik} \bG_{is} \right] \E \left[ \bG_{kl} \bG_{st} \right] x_l x_t  \nn \\
&=  \sum_{1 \leq l,t \leq n} \sum_{k \neq i, s \neq i} (\E [ \bU_{ik} \bU_{is}] + \bm_{ik}\bm_{is}) \E (\left[ \bU_{kl} \bU_{st} \right] + \bm_{kl} \bm_{st}) x_l x_t.
\end{align*}

We find the following useful expressions for each $i \in N$:
\begin{align*}
\sum_{1 \leq l,t \leq n} \sum_{k \neq i, s \neq i} \E [ \bU_{ik} \bU_{is}] \E \left[ \bU_{kl} \bU_{st} \right] x_l x_t &= \sum_{k \neq i} \sum_{l,t} \E [ \bU_{ik} \bU_{ik}] \E \left[ \bU_{kl} \bU_{kt} \right] x_l x_t = \sum_k \bv_{ik}^2 \langle \bx, \bB_k \bx \rangle, \\
\sum_{1 \leq l,t \leq n} \sum_{k \neq i, s \neq i} \bm_{ik}\bm_{is} \E \left[ \bU_{kl} \bU_{st} \right] x_l x_t &= \sum_{k}  \bm_{ik}^2 \langle \bx, \bB_k \bx \rangle, \\
\sum_{1 \leq l,t \leq n} \sum_{k \neq i, s \neq i}  \bm_{kl} \bm_{st} \E [ \bU_{ik} \bU_{is}] x_l x_t &= \sum_{l,t}  \langle x_l\bm^l, \bB_i (x_t\bm^t) \rangle = \langle \overline{\bG} \bx, \bB_i \overline{\bG} \bx \rangle, \\
\sum_{1 \leq l,t,k,s \leq n}  \bm_{ik}\bm_{is} \bm_{kl} \bm_{st}x_l x_t &= \sum_{l,t} (\overline{\bG}^2_{il} x_l) (\overline{\bG}^2_{it} x_t) = (\overline{\bG}^2 \bx)_i^2.
\end{align*}

Then, we obtain that for each $i \in N$,
\begin{align*}
\E \left[ \left( \sum_{1 \leq k,l \leq n} \bG_{ik} \bG_{kl} x_l \right)^2 \right] &= \sum_{k=1}^n (\bv_{ik}^2 + \bm_{ik}^2 ) \langle \bx, \bB_k \bx \rangle + \langle  \overline{\bG}\bx, \bB_i \overline{\bG}\bx \rangle + (\overline{\bG}^2 \bx)_i^2, \\
\E \left[ \sum_{j=1}^n \bG_{ij} x_j  \right] &=  \sum_{j=1}^n \bm_{ij} x_j = \langle \bm_i, \bx \rangle = (\overline{\bG} \bx)_i, \\
\E \left[ \sum_{1 \leq k,l \leq n} \bG_{ik} \bG_{kl} x_l \right] &= \sum_{k,l} \left( \bm_{ik} \bm_{kl} + \E \left[ \bU_{ik} \bU_{kl}\right] \right) x_l = (\overline{\bG}^2 \bx)_i + 0 = (\overline{\bG}^2 \bx)_i, \\
\E \left[  \sum_{1 \leq j,k,l \leq n} \bG_{ij}\bG_{ik} \bG_{kl} x_j x_l  \right] &= \sum_{1 \leq j,l \leq} \sum_{k \neq i} \E \left[\bG_{ij}\bG_{ik} \right] \E \left[ \bG_{kl} \right]  x_j x_l \\
    &= \sum_{1 \leq j,l,k \leq n} \left( \E \left[\bU_{ij}\bU_{ik}\right] + \bm_{ij}\bm_{ik} \right) \bm_{kl} x_j x_l \\
    &= \sum_{l=1}^n \left( \langle x_l \bm^l, \bB_i \bx \rangle + \langle x_l \bm^l, \bM_i \bx \rangle  \right) = \langle \overline{\bG}\bx, \bB_i \bx \rangle +  \langle \overline{\bG} \bx, \bM_i \bx \rangle.
\end{align*}
Thus, we combine the above expressions and obtain that for each $i \in N$,
\begin{align*}
& \E \left[ | (\bI + \bG + \bG^2)_i \bx - z_i |^2 \right]    \\
&=  \langle \bx, \bB_i \bx \rangle + \sum_{k=1}^n (\bv_{ik}^2 + \bm_{ik}^2 ) \langle \bx, \bB_k \bx \rangle + \langle  \overline{\bG}\bx, \bB_i \overline{\bG}\bx \rangle + 2\langle \overline{\bG}\bx, \bB_i \bx \rangle  + \langle \bx, \bM_i \bx \rangle + (\overline{\bG}^2 \bx)_i^2 \\
&\quad + 2 \langle \overline{\bG} \bx, \bM_i \bx \rangle +  2\langle x_i \bm_i, \bx \rangle + 2x_i(\overline{\bG}^2 \bx)_i +x_i^2   - 2 \langle z_i \bm_i, \bx \rangle - 2 z_i(\overline{\bG}^2 \bx)_i - 2z_i x_i + z_i^2.
\end{align*}

We now let $\alpha_k = \sum_{i=1}^n (\bv_{ik}^2 + \bm_{ik}^2)$. By summing over $i$, it follows that
\begin{align}\label{2ndobjective}
 \E \left[ \| (\bI \!+\! \bG \!+\! \bG^2) \bx \!-\! \bz \|^2 \right] \!=\! \sum_{i=1}^n \left(  \alpha_i \langle \bx,  \bB_i \bx \rangle \!+\!  \langle  (\overline{\bG}+ \bI) \bx, \bB_i (\overline{\bG} \!+\! \bI) \bx \rangle \right) + Q(\bx),
\end{align}
where $Q(\bx)$ is a quadratic function of $\bx$ not involving the uncertainty matrix $(\bB_i)_i$ for any $i \in N$. As such, we can consider $Q$ as an additional cost part of the DM's objective function. It is straightforward to see that the DM's objective function \eqref{2ndobjective} is still a linear function of $(\bB_i)_i$ as before. Consequently, we conclude that when $(\bB_1^*, \dots, \bB_n^*)$ is Nature's best response with respect to a given $\bx$ having no zero entry, its $i$th component $\bB_i^*$ is uniquely determined as a rank-1 matrix if and only if the following nondegeneracy condition holds:
\begin{align}
\alpha_i x_j x_k + \big((\overline{\bG} + \bI)\bx\big)_j  \big((\overline{\bG} + \bI)\bx\big)_k  \neq 0 \ \text{ for all } \, 1 \le j < k \le n.
\end{align}
Therefore, the proposition is proven.
\end{proof}

\end{document}